\documentclass[oribibl]{llncs}

\usepackage{version}

\excludeversion{ijcar} \includeversion{report} 

\usepackage[inline]{enumitem}
\usepackage{listings}
\usepackage{subcaption}
\usepackage{microtype}
\usepackage{times}
\usepackage{amsmath}
\usepackage{amsfonts}
\usepackage{amssymb}
\usepackage{amsthm}
\usepackage{thm-restate}
\usepackage{mathpartir} 
\usepackage{color}
\usepackage{xcolor}
\usepackage{xspace} 
\usepackage{graphicx}
\usepackage{fancybox}
\usepackage{url}
\usepackage{multirow}
\usepackage{colortbl}
\usepackage{mathabx}
\usepackage{booktabs}
\usepackage{tabularx}
\usepackage{fancyvrb}
\usepackage{array}

\usepackage{mdframed}
\mdfsetup{skipabove=0pt,skipbelow=5pt}

\usepackage{enumitem}
\setlist{nolistsep}

\definecolor{linkcol}{HTML}{5388C8}

\usepackage{hyperref}
\hypersetup{
  colorlinks = true,
  citecolor=purple,
  linkcolor=linkcol,
}

\usepackage{cleveref}

\VerbatimFootnotes

\newcommand{\conf}[1]{( #1 )}

\newcommand{\rem}[1]{}

\newcommand{\teq}{\approx}
\newcommand{\tneq}{\not\approx}

\newcommand{\ter}[1]{\mathcal{T}(#1)}

\newcommand{\sint}{\mathsf{Int}}
\newcommand{\sbool}{\mathsf{Bool}}
\newcommand{\sseq}{{\mathsf{Seq}}}
\newcommand{\sarray}{{\mathsf{Arr}}}
\newcommand{\selem}{{\mathsf{Elem}}}

\newcommand{\arrnth}{\seqnth_{\mathbf{A}}}
\newcommand{\arrupdate}{\sequpdate_{\mathbf{A}}}
\newcommand{\arrlen}[1]{{\mathbf \ell({#1})}}
\newcommand{\arrcontent}{{\mathbf c}}
\newcommand{\arrteq}{\teq_{\mathbf{A}}}

\newcommand{\doubleplus}{\ensuremath{\mathbin{+\mkern-5mu+}}}
\newcommand{\infixseqcon}{\doubleplus}
\newcommand{\seqempty}{\epsilon}
\newcommand{\sequnit}{\mathsf{unit}}
\newcommand{\seqconii}[2]{#1 \infixseqcon #2}
\newcommand{\seqconiii}[3]{#1 \infixseqcon #2 \infixseqcon #3}
\newcommand{\seqconiv}[4]{#1 \infixseqcon #2 \infixseqcon #3 \infixseqcon #4}
\newcommand{\seqconv}[5]{#1 \infixseqcon #2 \infixseqcon #3 \infixseqcon #4 \infixseqcon #5}

\newcommand{\sequpdate}{\mathsf{update}}
\newcommand{\seqextract}{\mathsf{extract}}
\newcommand{\seqnth}{\mathsf{nth}}
\newcommand{\seqlen}[1]{|#1|}

\newcommand{\ff}{\bot}

\newcommand{\Mo}{\mathbf{I}}

\newcommand{\Sc}{\mathsf{S}}

\newcommand{\Ac}{\mathsf{A}}

\newcommand{\Fc}{\mathsf{F}}

\newcommand{\Cc}{\mathsf{C}}

\newcommand{\nf}[1]{{#1}{\downarrow}}
\newcommand{\ror}{\quad \parallel \quad}

\newcommand{\basiccal}{\mathsf{BASE}}
\newcommand{\extcal}{\mathsf{EXT}}

\newcommand{\define}[1]{\textsl{#1}}

\newcommand{\unsat}{\ensuremath{\mathsf{unsat}}\xspace}
\newcommand{\sat}{\ensuremath{\mathsf{sat}}\xspace}

\newcommand{\cdclt}{CDCL$(T)$\xspace}

\newcommand{\smtlib}{{\small SMT-LIB}\xspace}
\newcommand{\cvc}{{\small cvc5}\xspace}

\newcommand{\zzz}{{\small Z3}\xspace}
\newcommand{\cpp}{{\small C++}\xspace}

\newcommand{\ziii}{\textsc{z}{\small 3}\xspace}

\newcommand{\seqtn}{\tname{Seq}}
\newcommand{\lian}{\tname{LIA}}

\newcommand{\sth}{T_\seqtn}
\newcommand{\lth}{T_\lian}
\newcommand{\ssig}{\Sigma_\seqtn}
\newcommand{\lsig}{\Sigma_{\lian}}

\newcommand{\bool}{\ensuremath{\mathsf{Bool}}\xspace}

\spnewtheorem{assumption}{Assumption}{\bfseries}{\itshape}
\Crefname{assumption}{\text{Assumption}}{\text{Assumptions}}
\crefname{assumption}{\text{Assumption}}{\text{Assumptions}}
\newcommand{\rn}[1]{\textsf{\small #1}}

\newcommand{\newpar}[1]{
\medskip
\noindent\textbf{#1}\;}

\renewcommand{\vec}[1]{\overline{\boldsymbol{#1}}}
\newcommand{\tname}[1]{\mathsf{#1}}

\newcommand{\ent}[1][]{\models_{#1}}
\newcommand{\notent}[1][]{\nvDash_{#1}}
\newcommand{\entnf}{\models_{\infixseqcon}}

\newcommand{\M}{\mathcal{M}}
\newcommand{\Mc}{\mathsf{M}}

\newcommand{\powerset}[1]{{\mathsf{P}}({#1})}

\newcommand{\arithmodel}{\mathcal{A}}

\newcommand{\sceq}[1]{\equiv_{#1}}
\newcommand{\scneq}[1]{\not\equiv_{#1}}

\newcommand{\cf}{{\alpha}}

\newcommand{\smallcap}{\vskip -0.5em}

\begin{document}

\pagestyle{plain}

\newcommand{\mytitle}{Reasoning About Vectors \\
using an SMT Theory of Sequences}

\hypersetup{
  pdfborder={0 0 0},
  linkcolor=blue,
  urlcolor=blue,
  citecolor=blue,
  pdftitle=\mytitle
}

\title{\mytitle\thanks{
This work was funded in part by the Stanford Center for Blockchain Research, NSF-BSF grant numbers 2110397 (NSF) and 2020704 (BSF), and Meta Novi.
Part of the work was done when the first author was an intern at Meta Novi.
}
}


\author{
\rm{Ying Sheng$^{\text{1}}$ \and
    Andres N\"otzli$^{\text{1}}$ \and
    Andrew Reynolds$^{\text{2}}$ \and
    Yoni Zohar$^{\text{3}}$ \and
    David Dill$^{\text{4}}$ \and\\
    Wolfgang Grieskamp$^{\text{4}}$ \and
    Junkil Park$^{\text{4}}$ \and
    Shaz Qadeer$^{\text{4}}$ \and 
    Clark Barrett$^{\text{1}}$ \and
    Cesare Tinelli$^{\text{2}}$}\\
}
\institute{
  {$^{\text{1}}$Stanford University\enskip $^{\text{2}}$The University of Iowa\enskip $^{\text{3}}$Bar-Ilan University\enskip $^{\text{4}}$Meta Novi}
}


\newcommand{\samelineand}{\qquad}

\maketitle

\vspace{-2em}
\begin{abstract}

Dynamic arrays, also referred to as vectors, are fundamental data structures used in many programs.
Modeling their semantics efficiently is crucial when reasoning about such programs.
The theory of arrays is 
widely supported but is not ideal, because the number of elements is fixed (determined by its index sort) and cannot be adjusted, which is a problem, given that the length of vectors often plays an important role when reasoning about vector programs.
In this paper, we propose reasoning about vectors using a theory of sequences.
We introduce the theory, propose a basic calculus adapted from one for the theory of strings, and extend it to efficiently handle common vector operations.
We prove that our calculus is sound and show how to construct a model when it terminates with a saturated configuration.
%
%
Finally, we describe an implementation of the calculus in \cvc and demonstrate its efficacy by evaluating it on verification conditions for smart contracts and benchmarks derived from existing array benchmarks.
%

\end{abstract}

\section{Introduction} \label{sec:intro}
Generic vectors are used in many programming 
languages. For example, in \cpp's standard library, they are provided by
\verb!std::vector!. 
%
Automated verification of software systems that manipulate vectors 
 requires an efficient and automated way of reasoning about them.
Desirable characteristics of any approach for reasoning about vectors include:
$(i)$~expressiveness---operations that are 
commonly performed on vectors should be supported;
$(ii)$~generality---vectors are
always ``vectors of'' some type (e.g., vectors of integers),
and so it is
desirable that vector reasoning be integrated within a more general framework; solvers for
satisfiability modulo theories (SMT) provide such a framework and are widely used 
in verification tools (see \cite{DBLP:reference/mc/BarrettT18} for a recent survey);
$(iii)$~efficiency---fast and efficient reasoning is essential for usability, especially as verification tools are increasingly used by non-experts and in continuous integration.

Despite the ubiquity of vectors in software 
on the one hand
and the effectiveness of SMT solvers for software verification on the other hand, 
there is not currently a clean way to represent vectors using operators from the \smtlib standard \cite{SMTLib2017}.
While the theory of arrays can be used, it is not a great fit because arrays have a fixed size determined by their index type.  Representing a dynamic array thus requires additional modeling work.  Moreover, to reach an acceptable level of expressivity, quantifiers are needed, which often makes the reasoning engine less efficient and robust.  
Indeed, part of the motivation for this work was frustration with array-based modeling in the Move Prover, a verification framework for smart contracts~\cite{ZCQ+20} (see~\Cref{sec:evaluation} for more information about the Move Prover and its use of vectors).
The current paper bridges this gap by studying and implementing a native theory of \emph{sequences} in the SMT framework, which satisfies the desirable properties for vector reasoning listed above.

We present two SMT-based calculi for determining satisfiability in the theory of sequences.
Since the decidability of even weaker theories is unknown
(see, e.g., \cite{DBLP:conf/tacas/BjornerTV09,DBLP:conf/hvc/GaneshMSR12}), we do not aim for a decision procedure.  Rather, we prove model and solution soundness (that is, when our procedure
terminates, the answer is correct).
Our first calculus leverages techniques for the theory of strings.
We generalize these techniques, lifting rules specific to string characters to more general rules for arbitrary element types. 
By itself, this base calculus is already quite effective. However, it misses
opportunities to perform high-level vector-based reasoning.  For example, both
reading from and updating a vector are very common operations in programming, and reasoning efficiently about the corresponding sequence operators is thus crucial.
Our second calculus addresses this gap by integrating
reasoning methods from array solvers (which handle reads and updates efficiently)
into the first procedure. 
Notice, however, that this integration is not 
trivial, as it must handle novel combinations of operators (such as the
combination of update and read operators with 
concatenation) as well as out-of-bounds cases that do not occur with 
ordinary arrays.
We have implemented both variants of our calculus in the \cvc SMT solver~\cite{DBLP:conf/tacas/BarbosaBBKLMMMN22} and evaluated them on benchmarks originating from the Move prover, as well as benchmarks that were translated from SMT-LIB array benchmarks.

As is typical, both of our calculi are agnostic to the sort of the elements in the sequence.
Reasoning about sequences of elements from a particular theory can then be done
via theory combination methods such as Nelson-Oppen \cite{NO79} or polite combination \cite{JBLPAR,RRZ05}.
The former can be done for stably infinite theories (and
the theory of sequences that we present here is stably infinite), while
the latter requires investigating the politeness of the theory, which we expect to do in future work.

The rest of the paper is organized as follows.
\Cref{sec:preliminaries} includes basic notions from first-order logic.
\Cref{sec:theory} introduces the theory of sequences and shows how it can be used to model vectors.
\Cref{sec:algorithms} presents calculi for this theory and discusses 
their correctness.
\Cref{sec:implementation} describes the implementation of these calculi in \cvc.
\Cref{sec:evaluation} presents an evaluation comparing several variations of the sequence solver in \cvc and \zzz.
We conclude in \Cref{sec:conclusion} with directions for further research.

\smallskip

\noindent
{\bf Related work:}
Our work crucially builds on a proposal by Bj{\o}rner et al.~\cite{bjorner2012smt}, but extends it in several key ways.  
First, their  implementation (for a logic they call \verb|QF_BVRE|) restricts 
the generality of the theory by allowing only bit-vector elements 
(representing characters) and assuming that sequences are bounded.  
In contrast, our calculus maintains full generality, allowing unbounded sequences and elements of arbitrary types. 
Second, while our core calculus focuses only on a subset of the operators in \cite{bjorner2012smt},
our implementation supports the remaining operators by reducing them to the core operators, and
also adds native support for the $\sequpdate$ operator, which is not included in~\cite{bjorner2012smt}.

The base calculus that we present for sequences builds on similar work 
for the theory of strings
~\cite{DBLP:conf/fmcad/BerzishGZ17,DBLP:conf/cav/LiangRTBD14}. 
We extend our base calculus to support array-like reasoning based on the weak-equivalence approach~\cite{DBLP:conf/frocos/ChristH15}.
Though there exists some prior work on extending the theory of arrays with more operators and reasoning about length~\cite{DBLP:journals/fmsd/AlbertiGP17,DBLP:conf/cav/EladRIKS21,DBLP:conf/vstte/FalkeMS13},
this work does not include support for most of the of the sequence operators
we consider here.

The SMT-solver \zzz \cite{DBLP:conf/tacas/MouraB08} also provides a solver for sequences.  However, its documentation is limited~\cite{programmingz3}, it does not support $\sequpdate$ directly, and its internal algorithms are not described in the literature.  
Furthermore, as we show in \Cref{sec:evaluation}, the performance of the \zzz implementation is generally inferior to our implementation in \cvc.
%
 

\section{Preliminaries}
\label{sec:preliminaries}
We assume the usual notions and terminology of many-sorted first-order logic with equality
(see, e.g., \cite{Enderton2001} for a complete presentation).
We consider many-sorted signatures $\Sigma$, each containing a set of sort symbols (including a Boolean sort \bool), a family of logical symbols $\teq$ for equality, with sort $\sigma \times \sigma \to \bool$ for all
sorts $\sigma$ in $\Sigma$ and interpreted as the identity relation, and a set of interpreted (and sorted) function symbols.
We assume the usual definitions of well-sorted terms, literals, and formulas as terms of sort \bool.
%
A literal is {\em flat} if it has the form
$\bot$,
$p(x_1,\ldots,x_n)$, $\neg p(x_1,\ldots, x_n)$, $x\teq y$, $\neg x \teq y$, or
$x\teq f(x_1,\ldots, x_n)$, where $p$ and $f$ are function symbols and
$x$, $y$, and $x_1,\ldots,x_n$ are variables.
A $\Sigma$-interpretation $\M$ is defined as usual,
satisfying $\M(\bot)=\mathrm{false}$ and assigns:
a 
set $\M(\sigma)$ to every sort $\sigma$ of $\Sigma$,
a function 
$\M(f):\M(\sigma_1)\times\ldots\times\M(\sigma_n)\rightarrow\M(\sigma)$ to any
function symbol $f$ of $\Sigma$ with arity
$\sigma_1\times\ldots\times\sigma_n\rightarrow\sigma$,
and an element $\M(x)\in\M(\sigma)$ to any
variable $x$ of sort $\sigma$.
The satisfaction relation between interpretations
and formulas is defined as usual and is denoted
by $\models$.

A \define{theory} is a pair $T = (\Sigma, \Mo)$, in which $\Sigma$ is a signature
and $\Mo$ is a class of $\Sigma$-interpretations,
closed under variable reassignment. 
The \define{models} of $T$ are the interpretations in $\Mo$
without any variable assignments.  A
$\Sigma$-formula $\varphi$ is \define{satisfiable} (resp.,
\define{unsatisfiable}) \define{in $T$} if it is satisfied by some (resp., no)
interpretation in $\Mo$.
Given a (set of) terms $S$, we write $\ter S$ to denote the set of all subterms of
$S$.
For a theory $T=(\Sigma,\Mo)$, a set $S$ of $\Sigma$-formulas and a $\Sigma$-formula 
$\varphi$,
we write $S\ent[T]\varphi$ if every 
interpretation $\M\in\Mo$ that satisfies
$S$ also satisfies $\varphi$.
By convention and unless otherwise stated, 
we use letters $w, x, y, z$ to denote variables and $s, t, u, v$ to denote terms.

The theory  $\lth=(\lsig,\Mo_{\lth})$ of {\em linear integer arithmetic} is based on the signature
$\lsig$ that includes a single sort $\sint$, all natural numbers as constant symbols, the unary $-$ symbol, the binary
$+$ symbol and the binary $\leq$ relation.
When $k\in\mathbb{N}$, we use the notation $k\cdot x$, inductively defined by
$0\cdot x=0$ and $(m+1)\cdot x=x+m\cdot x$.
In turn, $\Mo_{\lth}$ consists of all structures $\M$ for $\lsig$ in 
which the domain $\M(\sint)$ of $\sint$ is the set of integer numbers,
for every constant symbol $n\in\mathbb{N}$, $\M(n)=n$, and $+$, $-$, and $\leq$ are 
interpreted as usual.  
We use standard notation for integer intervals (e.g., $[a,b]$ for the set of integers $i$, where $a \le i \le b$ and $[a,b)$ for the set where $a \le i < b$).
\section{A Theory of Sequences}
\label{sec:theory}

\begin{figure}[t]

\small

\[
\begin{array}{l}
 
\begin{array}{@{}l@{\quad}l@{\quad}l@{\quad}l@{}}
\toprule
\textbf{Symbol} & \textbf{Arity} & \textbf{\smtlib} & \textbf{Description} \\
\midrule
 n & \sint & \text{n}& \text{All constants } \verb|n| \in \mathbb{N} \\
 + & \sint \times \sint \to \sint & \verb|+| & \text{Integer addition} \\
 - & \sint \to \sint  & \verb|-| & \text{Unary Integer minus} \\
 {\leq} & \sint \times \sint \rightarrow \sbool & \verb|<=| & \text{Integer inequality} \\
\midrule
 \seqempty & \sseq & \verb|seq.empty| & \text{The empty sequence} \\
 \sequnit & \selem \to \sseq & \verb|seq.unit| & \text{Sequence constructor} \\
  \seqlen{\_} & \sseq \to \sint & \verb|seq.len| & \text{Sequence length} \\
  \seqnth & \sseq \times \sint \to \selem & \verb|seq.nth| & \text{Element access} \\
  \sequpdate & \sseq \times \sint \times \selem \to \sseq & \verb|seq.update| & \text{Element update} \\
  \seqextract & \sseq \times \sint \times \sint \to \sseq & \verb|seq.extract| & \text{Extraction (subsequence)} \\
  \seqconiii{\_}{\cdots}{\_} & \sseq \times \cdots \times \sseq \to \sseq & \verb|seq.concat| & \text{Concatenation} \\
\bottomrule
\end{array}
\\[4ex]
\end{array}
\]
\caption{
  Signature for the theory of sequences.
}
\label{fig:sig}
\end{figure}



 


We define the theory $\sth$ of sequences.
Its signature 
$\ssig$ is given in \Cref{fig:sig}.
It includes the sorts $\sseq$, $\selem$, $\sint$, and $\sbool$,
intuitively denoting sequences, elements, integers, and Booleans,
respectively.
The first four lines include symbols of $\lsig$.
We write $t_1 \bowtie t_2$, with $\bowtie\ \in \{>, <, \leq\}$, as syntactic
sugar for the equivalent literal expressed using
$\leq$ (and possibly $\neg$).
The sequence symbols are given on the remaining lines.
Their arities are also given in \Cref{fig:sig}.
Notice that $\seqconiii{\_}{\cdots}{\_}$ is a variadic function symbol.

Interpretations $\M$ of $\sth$ interpret:
$\sint$ as the set of integers;
$\selem$ as some set;
$\sseq$ as the set of finite sequences whose elements are from $\selem$;
$\seqempty$ as the empty sequence; 
$\sequnit$ as a function that takes an element from $\M(\selem)$ and returns the sequence that contains only that element;
$\seqnth$ as a function that takes an element $s$ from $\M(\sseq)$ and an integer $i$
and returns the $i$th element of $s$, in case $i$ is non-negative
and is smaller than the length of $s$ (we take the first element of a sequence to have index $0$).
Otherwise, the function has no restrictions;
$\sequpdate$ as a function that takes an element $s$ from $\M(\sseq)$, an integer
$i$, and an element $a$ from $\M(\selem)$ and returns
the sequence obtained from
$s$ by replacing its $i$th element by $a$, in case
$i$ is non-negative and smaller than the length of $s$.
Otherwise, the returned value is $s$ itself;
$\seqextract$ as a function that takes a sequence $s$ and
integers $i$ and $j$, and returns the maximal sub-sequence of 
$s$ that starts at index $i$ and has length at most $j$,
in case both $i$ and $j$ are non-negative and $i$ is smaller than the
length of $s$. Otherwise, the returned value is the empty sequence;\footnote{In \cite{bjorner2012smt}, the second argument $j$ denotes the end index, while here
it denotes the length of the sub-sequence, in order to be 
consistent with the theory of strings in the \smtlib standard.}
$\seqlen{\_}$ as a function that takes a sequence and returns its length; and
$\seqconiii{\_}{\cdots}{\_}$ as a function that takes some number of sequences (at least 2)
and returns their concatenation.

Notice that the interpretations of $\selem$ and $\seqnth$ are not completely fixed by the theory:
$\selem$ can be set arbitrarily, and $\seqnth$ is only defined by the theory for some values of
its second argument. For the rest, it can be set arbitrarily.


\subsection{Vectors as Sequences}
We show the applicability of $\sth$ 
by using it for 
a simple verification task.
Consider the \cpp function \verb!swap! at the top of \Cref{fig:ex_seq_arr}.
This function swaps two elements in a vector. 
The comments above the function include a partial specification for it:
if both indexes are in-bounds and the indexed elements are equal, then the function should not change 
the vector (this is expressed by \verb!s_out==s!). 
%
We now consider how to encode the verification condition induced by the code and the specification.
The function variables $a$, $b$, $i$, and $j$ can be encoded
 as variables of sort $\sint$ with the same names.
We include two copies of $s$:
$s$ for its value at the beginning, and $s_{out}$ for
its value at the end.
But what should be the sorts of $s$ and $s_{out}$?
In \Cref{fig:ex_seq_arr} we consider two options: one is based on arrays and the other on sequences.

\begin{example}[Arrays]
\label{ex:arr-enc}
The theory of arrays includes three sorts:
index, element (in this case, both are $\sint$), and an array sort $\sarray$, as well as two operators:
$x[i]$, interpreted as the $i$th element of $x$;
and $x[i\leftarrow a]$, interpreted as the array obtained from
$x$ by setting the element at index $i$ to $a$.
%
We declare $s$ and $s_{out}$ as variables of an uninterpreted sort $V$ and declare two functions
$\ell$ and $\arrcontent$, which, given $v$ of sort $V$,
return its length
(of sort $\sint$)
and content (of sort $\sarray$), respectively.\footnote{It is possible to obtain a similar encoding using the theory of datatypes; however, here we use uninterpreted functions which are simpler and better supported by SMT solvers.}
%

Next, we introduce functions to model vector operations: $\arrteq$ for comparing vectors,
$\arrnth$ for reading from them, and $\arrupdate$ for updating them. 
These functions need to be axiomatized. 
We include two axioms (bottom of
\Cref{fig:ex_seq_arr}):
$Ax_1$ states that two vectors are equal iff they have the same length and the same contents.
$Ax_2$ axiomatizes the update operator; the result has the same length, and if the updated index is
in bounds, then the corresponding element is updated. 
These axioms are not meant to be complete, but are rather just strong enough for the example.

The first two lines of the \verb!swap! function are encoded as equalities using $\arrnth$,
and the last two lines are combined into one nested constraint that involves $\arrupdate$.
The precondition of the specification is naturally modeled
using $\arrnth$, and the post-condition is negated, so that
the unsatisfiability of the formula entails the 
correctness of the function w.r.t. the specification. 
Indeed, the conjunction of all formulas in this encoding is unsatisfiable in the combined theories
of arrays, integers, and uninterpreted functions.
\end{example}

The above encoding
 has two main shortcomings:
It introduces auxiliary symbols, and it uses quantifiers, thus reducing clarity and efficiency. 
%
%
 In the next example, we see how using the theory of sequences
 allows for a much more natural and succinct encoding.

\begin{example}[Sequences]
\label{ex:seq-enc}
In the sequences encoding, $s$ and $s_{out}$
have sort $\sseq$.
No auxiliary sorts or functions
are needed, as the theory symbols can be used directly.
Further, these symbols do not need to be axiomatized as
their semantics is fixed by the theory.
The resulting formula, much shorter than in \Cref{ex:seq-enc}
and with no quantifiers, is unsatisfiable in $\sth$.
\end{example}


\begin{figure}[t]
    \centering
\begin{lstlisting}[language=C++,basicstyle=\scriptsize]
// @pre: 0 <= i,j < s.size() and s[i] == s[j]
// @post: s_out == s
void swap(std::vector<int>& s, int i, int j) {
  int a = s[i];
  int b = s[j];
  s[i] = b;
  s[j] = a;
}
\end{lstlisting}
    
\scalebox{0.85}{
\setlength{\tabcolsep}{8pt}
\begin{tabular}{@{}lll@{}}
\toprule
&  \textbf{Sequences} &  \textbf{Arrays} \\
\midrule
Problem Variables & $a,b,i,j: \sint$\qquad $s,s_{out}: \sseq$ & 
$a,b,i,j: \sint$\qquad $s,s_{out}: V$ \\
\midrule
Auxiliary Variables &&
$\ell: V\rightarrow \sint$ \qquad
$\arrcontent: V\rightarrow \sarray$ \\
&&
$\arrteq: V\times V\rightarrow \bool$ \\
&&$\arrnth: V\times \sint\rightarrow \sint $ \\
&&
$\arrupdate: V\times \sint \times \sint \rightarrow V$ \\
\midrule
Axioms && $Ax_{1}\wedge Ax_{2}$\\
\midrule
Program &    $a\teq\seqnth(s,i)\wedge b\teq \seqnth(s,j)$&  $a\teq\arrnth(s,i)\wedge b\teq \arrnth(s,j)$\\
&    $s_{out}\teq\sequpdate(\sequpdate(s, i, b), j, a)$ & $s_{out}\arrteq\arrupdate(\arrupdate(s,i,b), j, a)$\\
\midrule
Spec. &    $0\leq i,j < \seqlen{s}\wedge \seqnth(s,i)\teq\seqnth(s,j)$ & $0\leq i,j < \arrlen{s}\wedge \arrnth(s,i)\teq\arrnth(s,j)$\\[1ex]
 &   $\neg s_{out}\teq s$ & $\neg s_{out}\arrteq s$ \\
\bottomrule
\end{tabular}
}

\medskip

\scalebox{0.9}{
\begin{tabular}{|c|}
\hline
$Ax_{1}:= \forall x,y.x\arrteq y \leftrightarrow 
(\arrlen{x}\teq\arrlen{y}\wedge\forall~0\leq i < 
\arrlen{x}.\arrcontent(x)[i]\teq\arrcontent(y)[i])$ \\[1ex]
$Ax_{2}:= \forall x,y,i,a.y\arrteq \arrupdate(x,i,a)\rightarrow (\arrlen{x}\teq \arrlen{y}\wedge (0\leq i < \arrlen{x}\rightarrow \arrcontent(y)\teq\arrcontent(x)[i\leftarrow a] ))$
\\\hline
\end{tabular}
}
\normalsize
{\vskip -0.5em}
    \caption{An example using $\sth$.}
    
    \label{fig:ex_seq_arr}
\end{figure}

\section{Calculi}
\label{sec:algorithms}

%
After introducing some definitions and assumptions, we describe a basic
calculus for the theory of sequences, which adapts techniques from previous
procedures for the theory of strings.
In particular, the basic calculus reduces the operators $\mathsf{nth}$ and $\mathsf{update}$ by introducing concatenation terms.
We then show how to extend the basic calculus by introducing additional rules inspired by solvers for the theory of arrays; the modified calculus can often reason about $\mathsf{nth}$ and $\mathsf{update}$ terms directly, avoiding the introduction of concatenation terms (which are typically expensive to reason about).

Given a vector of sequence terms $\vec{t} = (t_1, \ldots, t_n)$,
we use $\vec{t}$ 
to denote the term corresponding to the concatenation of $t_1, \ldots, t_n$.
If $n=0$,  $\vec{t}$ denotes $\seqempty$, and
if $n=1$, $\vec{t}$ denotes $t_1$; otherwise (when $n>1$),
$\vec{t}$ denotes a concatenation term having $n$ children.
In our calculi, we distinguish between sequence and arithmetic constraints.

\begin{definition}
A $\ssig$-formula 
$\varphi$ is a \define{sequence constraint}
if it has the form
	$s\teq t$ or $s\tneq t$;
it is an \define{arithmetic constraint} if it has the form
$s\teq t$, $s\geq t$, $s\tneq t$, or $s<t$
where $s,t$ are terms of sort $\sint$, or if it is a
disjunction $c_1 \lor c_2$ of two arithmetic constraints.
\end{definition}

\noindent Notice that sequence constraints do not have to contain sequence terms
(e.g., $x\teq y$ where $x,y$ are $\selem$-variables).
Also, equalities and disequalities between terms of sort $\sint$ are both sequence and arithmetic constraints.
In this paper we focus on sequence constraints and arithmetic
constraints. 
This is justified by the following lemma.
(Proofs of this lemma and later results can be found in
\begin{report}%
the appendix%
\end{report}%
\begin{ijcar}%
an extended version of this paper~\cite{ShengEtAl-arXiv-22}%
\end{ijcar}%
.)

\begin{restatable}{lemma}{flatas}
\label{lem:flatas}
For every quantifier-free $\ssig$-formula $\varphi$,
there are sets $\Sc_1,\ldots,\Sc_n$ of sequence constraints
and sets $\Ac_1,\ldots,\Ac_n$ of arithmetic constraints
such that
$\varphi$ is $\sth$-satisfiable 
iff $\Sc_i\cup\Ac_i$ is $\sth$-satisfiable for some $i\in[1,n]$.
\end{restatable}

\noindent 
Throughout the presentation of the calculi, we will make a few simplifying assumptions.

\begin{assumption}
\label{assumption}
  Whenever we refer to a set $\Sc$ of sequence constraints, we assume:
\begin{enumerate}
    \item 
  for every non-variable term $t\in\ter{\Sc}$, there exists a variable $x$ such that $x\teq
  t\in\Sc$;
  \item for every $\sseq$-variable $x$, there exists a variable $\ell_x$ such that
  $\ell_x\teq\seqlen{x}\in\Sc$;
  \item all literals in $\Sc$ are flat.
  \end{enumerate}
  Whenever we refer to a set of arithmetic constraints, we assume all its literals are flat.
\end{assumption}
These assumptions are without loss of generality as any set 
can easily be transformed into an equisatisfiable set satisfying the assumptions by the addition of fresh
variables and equalities. 
Note that some rules below introduce non-flat literals.  In such cases, we assume that similar transformations are done immediately after applying the rule to maintain the invariant that all literals in $\Sc\cup\Ac$ are flat.
Rules may also introduce fresh variables $k$ of sort $\sseq$.
We further assume that in such cases, a corresponding 
constraint $\ell_{k}\teq\seqlen{k}$ is added to $\Sc$ with
a fresh variable $\ell_k$.

\begin{definition}
\label{def:entempty}
Let $\Cc$ be a set of constraints. 
We write $\Cc \ent \varphi$ to denote that
$\Cc$ entails formula $\varphi$ in the empty theory, and write
$\sceq{\Cc}$ to denote the binary relation over $\ter{\Cc}$ 
such that
$s\sceq{\Cc}t$ iff $\ \Cc\ent s\teq t$.
\end{definition}

\begin{lemma}
\label{lem:equiv_class_nonempty}
For all set $\Sc$ of sequence constraints,
$\sceq{\Sc}$ is an equivalence relation;
furthermore,
  every equivalence class of $\sceq{\Sc}$ contains at least
  one variable.
\end{lemma}

\noindent
We denote the equivalence class of a term $s$ according to $\sceq{\Sc}$
by $[s]_{\sceq{\Sc}}$ and 
drop the $\sceq{\Sc}$ subscript when it is clear from the context.

\begin{figure}[t]
\centering
\setlength{\tabcolsep}{12pt}
\begin{tabular}{@{}ll@{}}
     $\seqlen{\seqempty} \to 0$ & $\seqlen{\sequnit(t)} \to 1$ \\
     $\seqlen{\sequpdate(s, i, t)} \to \seqlen{s}$ & $\seqlen{\seqconiii{s_1}{\cdots}{s_n}} \to \seqlen{s_1} + \cdots + \seqlen{s_n}$ \\[2ex]
     $\seqconiii{\vec{u}}{\seqempty}{\vec{v}} \to \seqconii{\vec{u}}{\vec{v}}$ & $\seqconiii{\vec{u}}{(\seqconiii{s_1}{\cdots}{s_n})}{\vec{v}} \to \seqconv{\vec{u}}{s_1}{\cdots}{s_n}{\vec{v}}$ \\
\end{tabular}
{\vskip -0.5em}
\caption{Rewrite rules for the reduced form $\nf{t}$ of a term $t$,
obtained from $t$ by applying these rules to completion.}
\label{fig:nf}
\end{figure}

In the presentation of the calculus, it will often be useful to normalize terms to what will be called a \emph{reduced form}.

\begin{definition}
\label{def:nf}
  Let $t$ be a $\ssig$-term.
  The \emph{reduced form} of $t$, denoted by $\nf{t}$, is
  the term obtained by applying the rewrite rules listed in
  \Cref{fig:nf} to completion.
\end{definition}

\noindent 
Observe that $\nf{t}$ is well defined
because the given rewrite rules form a terminating rewrite system.
This can be seen by noting that 
each rule reduces the number of applications of sequence operators
in the left-hand side term
or keeps that number the same but reduces the size of the term.
It is not difficult to show that $\ent[\sth] t \teq \nf{t}$.
%


We now introduce some basic definitions related to concatenation terms.

\begin{definition}
\label{def:atomic}
A \emph{concatenation term} is a term of the form
$\seqconiii{s_1}{\cdots}{s_n}$ with $n\geq 2$.  
If each $s_i$ is a variable, it is a \emph{variable concatenation term}.
For a set $\Sc$ of sequence constraints,
%
a variable concatenation term $\seqconiii{x_1}{\cdots}{x_n}$
is \emph{singular} in $\Sc$
if $\Sc \not\ent x_i \teq \seqempty$ for at most one variable $x_i$
with $i \in [1,n]$.
A sequence variable $x$ is \emph{atomic in $\Sc$}
if
$\Sc\notent x\teq\seqempty$ and
for all variable concatenation terms $s\in\ter{\Sc}$ such that $\Sc\ent x\teq s$,
$s$ is singular in $\Sc$.
\end{definition}

\noindent We lift the concept of atomic variables to atomic representatives
of equivalence classes. 

\begin{definition}
\label{def:equiv-represent}
Let $\Sc$ be a set of sequence constraints.
Assume a choice function $\cf: \ter{\Sc} /{\sceq{\Sc}} \rightarrow
\ter{\Sc}$ that chooses a variable from each equivalence class of $\sceq{\Sc}$.
A sequence variable $x$
is an \emph{atomic representative in $\Sc$} if it is atomic in $\Sc$ and 
$x=\cf([x]_{\sceq{\Sc}})$.
\end{definition}

\noindent Finally, we introduce a relation that is the foundation for reasoning
about concatenations.

\begin{definition}
\label{def:entnf}
Let $\Sc$ be a set of sequence constraints.  We inductively define a relation $\Sc \entnf x \teq s$, where $x$ is a sequence variable in $\Sc$ and $s$ is a sequence term whose variables are in $\ter{\Sc}$, as follows:
\begin{enumerate}
    \item\label{it:entnf-ref} $\Sc \entnf x \teq x$ for all sequence variables $x\in\ter{\Sc}$.
    \item\label{it:entnf-con} $\Sc \entnf x \teq t$ for all sequence variables $x\in\ter{\Sc}$ and variable concatenation terms $t$, where $x\teq t\in\Sc$.
    \item\label{it:entnf-rec} If $\Sc \entnf x \teq \nf{(\seqconiii{\vec{w}}{y}{\vec{z}})}$ and $\Sc \ent y \teq t$ and
    $t$ is $\seqempty$ or a variable concatenation term in $\Sc$ that is not singular in $\Sc$,
    then $\Sc \entnf x \teq \nf{(\seqconiii{\vec{w}}{t}{\vec{z}})}$.
\end{enumerate}
Let $\cf$ be a choice function for $\Sc$ as defined in \Cref{def:equiv-represent}. We additionally define the entailment relation $\Sc \entnf^\ast x \teq \vec{y}$, where $\vec{y}$ is of length $n \ge 0$,
to hold if each element of $\vec{y}$ is an atomic representative in $\Sc$
and there exists $\vec{z}$ of length $n$ such that
$\Sc \entnf x \teq \vec{z}$ and
$\Sc \ent y_i \teq z_i$ for $i \in [1,n]$.
\end{definition}

\noindent In other words, $\Sc \entnf^\ast x \teq t$ holds when $t$ is a
concatenation of atomic representatives and is entailed to be equal to $x$
by $\Sc$.  
In practice, $t$ is determined by
recursively expanding concatenations using equalities in $\Sc$ until a fixpoint
is reached.

\begin{example}
\label{ex:entnf}
  Suppose $\Sc=\{x\teq \seqconii{y}{z},y\teq \seqconii{w}{u},u \teq v\}$ (we omit the additional constraints required by \Cref{assumption}, part 2 for brevity).  It is easy to see that $u$, $v$, $w$, and $z$ are atomic in $\Sc$, but $x$ and $y$ are not.  Furthermore, $w$ and $z$ (and one of $u$ or $v$) must also be atomic representatives.
Clearly, $\Sc\entnf x\teq x$ and
$\Sc\ent x\teq\seqconii{y}{z}$.
Moreover, $\seqconii{y}{z}$ is a variable concatenation term that is not singular in $\Sc$.
Hence, we have
$\Sc\entnf x\teq\nf{( \seqconii{y}{z})}$, and so
$\Sc\entnf x\teq \seqconii{y}{z}$ (by using either \Cref{it:entnf-con} or \Cref{it:entnf-rec} of \Cref{def:entnf}, 
as in fact $x\teq\seqconii{y}{z}\in\Sc$.
	).
Now, since
$\Sc\entnf x\teq \seqconii{y}{z}$,
$\Sc\ent y\teq \seqconii{w}{u}$, and
$\seqconii{w}{u}$ is a variable concatenation term not singular
in $\Sc$, we get that
$\Sc\entnf x\teq \nf{(\seqconii{(\seqconii{w}{u})}{z})}$,
and so
$\Sc\entnf x\teq \seqconiii{w}{u}{z}$.
Now, assume that
$v=\alpha([v]_{\sceq{\Sc}})=\alpha(\{v,u\})$.
Then,
$\Sc\entnf^{\ast}x\teq\seqconiii{w}{v}{z}$.
\end{example}



\noindent 
Our calculi can be understood as modeling abstractly a cooperation 
between an \emph{arithmetic subsolver} and a \emph{sequence subsolver}.
Many of the derivation rules lift those in the string calculus of Liang et al.~\cite{DBLP:conf/cav/LiangRTBD14} to sequences of elements of an arbitrary type.  
We describe them similarly as rules that modify \emph{configurations}.

\begin{definition}
A \emph{configuration} is either the distinguished configuration $\unsat$ or
a pair $\conf{\Sc,\Ac}$ of a set $\Sc$ of sequence constraints and
a set $\Ac$ of arithmetic constraints.
\end{definition}

\noindent The rules are given in \define{guarded assignment form}, 
where 
the rule premises describe the conditions on the current configuration
under which the rule can be applied,
and the conclusion is either $\unsat$,
or otherwise describes the resulting modifications to the configuration.
A rule may have multiple conclusions separated by $\parallel$.
In the rules, some of the premises have the form $\Sc\ent[]s\teq t$
(see \Cref{def:entempty}).  
Such entailments can be checked with standard algorithms for congruence closure.
Similarly, premises of the form $\Sc\ent[\lian]s\teq t$ can be checked
by solvers for linear integer arithmetic.

An application of a rule is \define{redundant} if it has a conclusion
where each component in the derived configuration is a subset of 
the corresponding component in the premise configuration.  We assume that for rules that introduce fresh variables, the introduced variables are identical whenever the premises triggering the rule are the same (i.e., we cannot generate an infinite sequence of rule applications by continuously using the same premises to introduce fresh variables).\footnote{In practice, this is implemented by associating each introduced variable with a \emph{witness term} as described in~\cite{DBLP:conf/fmcad/ReynoldsNBT20}.}
%
%
A configuration other than \unsat is \define{saturated} with respect to 
a set $R$ of derivation rules
if every possible application of a rule in $R$ to it is redundant.
A \define{derivation tree} is a tree where each node is a configuration 
whose children, if any, 
are obtained by a non-redundant application of a rule of the calculus.
A derivation tree is \define{closed} if all of its leaves are $\unsat$.
As we show later,
a closed derivation tree with root node $\conf{ \Sc, \Ac }$ is a proof that
$\Ac \cup \Sc$ is unsatisfiable in  $\sth$.
In contrast, a derivation tree with root node $\conf{ \Sc, \Ac }$ and a
saturated leaf with respect to all the rules of the calculus
is a witness that $\Ac \cup \Sc$ is satisfiable in $\sth$.

\subsection{Basic Calculus}


%

\begin{figure}[t]
\centering
\scalebox{.85}{
\begin{tabular}{cc}
\rn{A-Conf}
\(
\inferrule{
  \Ac \ent[\lian] \ff
}{
 \unsat
}
\)
\qquad
\rn{A-Prop}
\(
\inferrule{
  \Ac \ent[\lian] s \teq t
  \\
  s, t \in \ter{\Sc}
}{
  \Sc := \Sc, s \teq t
}
\)
\\[4ex]
\rn{S-Conf}
\(
\inferrule{
 \Sc \ent \ff
}{
 \unsat
}
\)
\qquad
\rn{S-Prop}
\(
\inferrule{
  \Sc \ent s \teq t \\ s,t\in\ter{\Sc} \\
  s,t\text{ are } \lsig\text{-terms}
}{
  \Ac := \Ac, s \teq t
}
\)
\\[4ex]
\rn{S-A}
\(
\inferrule{
  x,y\in\ter{\Sc}\cap\ter{\Ac} \\
  x,y:\sint
}{
  \Ac := \Ac, x \teq y\ror \Ac:=\Ac,x\tneq y
}
\)
\\[4ex]
\rn{L-Intro}
\(
\inferrule{
  s \in \ter{\Sc}
  \\
  s : \sseq
}{
  \Sc := \Sc, \seqlen{s} \teq \nf{(\seqlen{s})}
}
\)
\quad
\rn{L-Valid}
\(
\inferrule{
  x \in \ter{\Sc}
  \\
  x : \sseq
}{
  \Sc := \Sc, x \teq \seqempty
  \ror
  \Ac := \Ac, \ell_x > 0
}
\)
\\[4ex]
\rn{U-Eq}
\(
\inferrule{
  \Sc \ent \sequnit(x) \teq \sequnit(y)
}{
  \Sc := \Sc, x \teq y
}
\)
\quad
\rn{C-Eq}
\(
\inferrule{
  \Sc \entnf^\ast x \teq \vec{z}
  \\
  \Sc \entnf^\ast y \teq \vec{z}
}{
  \Sc := \Sc, x \teq y
}
\)
%
%
\\[4ex]
\rn{C-Split}
\(
\inferrule{
  \Sc \entnf^\ast x \teq \nf{(\seqconiii{\vec{w}}{y}{\vec{z}})}
  \\
  \Sc \entnf^\ast x \teq \nf{(\seqconiii{\vec{w}}{y'}{\vec{z'}})}
}{
  \Ac := \Ac, \ell_{y} > \ell_{y'} \\ \Sc := \Sc, y \teq \seqconii{y'}{k}
  \ror \\\\
  \Ac := \Ac, \ell_{y} < \ell_{y'} \\ \Sc := \Sc, y' \teq \seqconii{y}{k}
  \ror \\\\
  \Ac := \Ac, \ell_{y} \teq \ell_{y'} \\ \Sc := \Sc, y \teq y'
  \\\\
}
\)
\\[4ex]
\rn{Deq-Ext}
\(
\inferrule{
 x \tneq y\in\Sc \\ x, y : \sseq
}{
  \Ac := \Ac, \ell_{x} \tneq \ell_{y}
  \ror \\\\
  \Ac := \Ac, \ell_{x} \teq \ell_{y}, 0 \leq i < \ell_{x} \quad \Sc := \Sc, w_1\teq\seqnth(x, i),w_2\teq \seqnth(y, i), w_1\tneq w_2
}
\)
\end{tabular}}
{\vskip -0.5em}
\caption{
Core derivation rules.
The rules use $k$ and $i$ to denote fresh variables of sequence and integer sort, respectively, and $w_1$ and $w_2$ for fresh element variables.
}
\label{fig:core-rules}
\end{figure}

\begin{figure}
\centering
\scalebox{.85}{
\begin{tabular}{cc}
\\[4ex]
\rn{R-Extract}
\(
\inferrule{
  x\teq\seqextract(y,i,j)\in\Sc
}{
  \Ac := \Ac, i<0\vee i\geq \ell_{y}\vee j\leq 0 \\ \Sc := \Sc, x\teq\seqempty
  \ror \\
  \Ac := \Ac, 0\leq i < \ell_{y}, j > 0, \ell_{k}\teq i,
\ell_{x} \teq \min(j,\ell_{y}-i)\\\\
  \Sc := \Sc, y\teq \seqconiii{k}{x}{k'}
}
\)
\\[8ex]
\rn{R-Nth}
\(
\inferrule{
  x \teq \seqnth(y, i) \in \Sc
}{
  \Ac := \Ac, i < 0 \vee i \geq \ell_{y} 
  \ror \\\\
  \Ac := \Ac, 0 \leq i < \ell_{y}, \ell_{k} \teq i \\ \Sc := \Sc, y \teq \seqconiii{k}{\sequnit(x)}{k'}
  \\\\
}
\)
\\[4ex]
\rn{R-Update}
\(
\inferrule{
  x \teq \sequpdate(y, i, z) \in \Sc
}{
  \Ac := \Ac, i < 0 \vee i \geq \ell_{y} \\ \Sc := \Sc, x\teq y
  \ror \\\\
  \Ac := \Ac, 0 \leq i < \ell_{y}, \ell_{k} \teq i, \ell_{k'} \teq 1
  \\
  \Sc := \Sc, y \teq \seqconiii{k}{k'}{k''}, x \teq \seqconiii{k}{\sequnit(z)}{k''}
}
\)
\end{tabular}}
\caption{
Reduction rules for $\seqextract$, $\seqnth$, and $\sequpdate$.  The rules use $k$, $k'$, and $k''$ to denote fresh sequence variables.  We write $s\teq\min(t,u)$ as an abbreviation for $s\teq t \vee s\teq u, s\le t, s\le u$.
}
\label{fig:reduction-rules}
\end{figure}

\begin{definition}
\label{def:basiccal}
The calculus $\basiccal$ consists of the derivation rules in
\Cref{fig:core-rules,fig:reduction-rules}.
\end{definition}
Some of the rules are adapted from previous work on string
solvers~\cite{DBLP:conf/cav/LiangRTBD14, DBLP:conf/cav/ReynoldsWBBLT17}.
Compared to that work, our presentation of the rules is noticeably simpler,
due to our use of the relation $\entnf^{\ast}$ 
from \Cref{def:entnf}.
In particular, our configurations consist only of pairs of sets of formulas, without
any auxiliary data-structures. 

Note that judgments of the form
$\Sc\entnf^{\ast} x\teq t$ are used in premises 
of the calculus.
It is possible to compute whether such a premise holds
thanks to the following lemma.

\begin{restatable}{lemma}{entnfterm}
\label{lem:entnf-term}
Let $\Sc$ be a set of sequence constraints
and $\Ac$ a set of arithmetic constraints.
If $\conf{\Sc,\Ac}$ is saturated w.r.t.
$\rn{S-Prop}$,
$\rn{L-Intro}$
and $\rn{L-Valid}$,
the problem of determining whether 
$\Sc\entnf^{\ast}x\teq s$ for given $x$ and $s$
is decidable.
\end{restatable}

\noindent
\Cref{lem:entnf-term} assumes saturation with respect to certain rules. 
Accordingly, our proof strategy,  described in \Cref{sec:implementation},
will ensure such saturation before attempting to apply rules relying on $\entnf^{\ast}$.
The relation $\entnf^{\ast}$ induces a normal form
for each equivalence class of $\sceq{\Sc}$.

\begin{restatable}{lemma}{nfunique}
\label{lem:nf-unique}
Let $\Sc$ be a set of sequence constraints
and $\Ac$ a set of arithmetic constraints.
Suppose $\conf{\Sc,\Ac}$ is saturated w.r.t.
$\rn{A-Conf}$,
$\rn{S-Prop}$,
$\rn{L-Intro}$,
$\rn{L-Valid}$, and
$\rn{C-Split}$.
Then, 
    for every equivalence class $e$ of $\sceq{\Sc}$ whose terms are of sort $\sseq$, there 
    exists a unique (possibly empty) $\vec{s}$
    such that whenever $\Sc \entnf^\ast x \teq \vec{s'}$ for $x \in e$, then
    $\vec{s'} =\vec{s}$.
    In this case, we call $\vec{s}$
    the \emph{normal form} of $e$ (and of $x$).
\end{restatable}


%
We now turn to the description of the rules in \Cref{fig:core-rules},
which form the core of the calculus.
For greater clarity, some of the conclusions of the rules include terms before they are flattened.
First, either subsolver can report that the current set of constraints is
unsatisfiable by using the rules \rn{A-Conf} or \rn{S-Conf}.
For the former, the entailment $\ent[\lian]$ 
(which abbreviates $\ent[\lth]$)
can be checked by a standard
procedure for linear integer arithmetic, and the latter corresponds to a situation where congruence closure detects a conflict between an equality and a disequality.
The rules \rn{A-Prop}, \rn{S-Prop}, and \rn{S-A} correspond to a form of
Nelson-Oppen-style theory combination between the two sub-solvers.
The first two communicate equalities between the sub-solvers, while the third
guesses arrangements for shared variables of sort $\sint$.
%
%
\rn{L-Intro} ensures that the length term $\seqlen{s}$ for each sequence term $s$ is equal to its reduced form $\nf{(\seqlen{s})}$.
\rn{L-Valid} restricts sequence lengths to be non-negative, splitting on whether each sequence is empty or has a length greater than $0$.
%
%
The $\sequnit$ operator is injective, which is captured by \rn{U-Eq}.
\rn{C-Eq} concludes that two sequence terms are equal if they have the same normal form.
If two sequence variables have different normal forms, then \rn{C-Split} takes the first differing components $y$ and $y'$ from the two normal forms and splits on their length relationship.
Note that \rn{C-Split} is the source for non-termination of the calculus (see, e.g., \cite{DBLP:conf/cav/LiangRTBD14, DBLP:conf/cav/ReynoldsWBBLT17}).
Finally, \rn{Deq-Ext} handles disequalities between sequences $x$ and $y$ by either asserting that their lengths are different or by choosing an index $i$ at which they differ.

\Cref{fig:reduction-rules} includes a set of reduction rules
for handling operators that are not directly handled by the core rules.
These reduction rules capture the semantics of these operators by reduction
to concatenation.
\rn{R-Extract} 
splits into two cases: Either the extraction uses an out-of-bounds
index or a non-positive length, in which case the result is the empty sequence,
or the original sequence can be described as a concatenation that includes
the extracted sub-sequence.
\rn{R-Nth} creates an equation between $y$ and a concatenation term
with $\sequnit(x)$ as one of its components, as long as $i$ is not
out of bounds.
\rn{R-Update} considers two cases. If $i$ is out of bounds, then the update term is equal to $y$.  Otherwise, $y$ is equal to a concatenation, with the middle component ($k'$) representing the part of $y$ that is updated. In the $\sequpdate$ term,  $k'$ is replaced by $\sequnit(z)$.

\begin{example}
  Consider a configuration $\conf{\Sc,\Ac}$,
  where $\Sc$ contains the formulas 
  $x \teq \seqconii{y}{z}$, 
  $z  \teq \seqconiii{v}{x}{w}$,
  and 
  $v \teq \sequnit(u)$,
  and $\Ac$ is empty.
  Hence, $\Sc\ent \seqlen{x}\teq\seqlen{\seqconii{y}{z}}$.
  By $\rn{L-Intro}$, we have
  $\Sc\ent\seqlen{\seqconii{y}{z}}\teq\seqlen{y}+\seqlen{z}$.
  Together with \Cref{assumption}, we have
  $\Sc\ent\ell_x\teq\ell_y+\ell_z$, and then with $\rn{S-Prop}$, we have  
$\ell_x\teq\ell_y+\ell_z\in\Ac$.
  Similarly,  we can derive
  $\ell_z\teq \ell_v+\ell_x+\ell_w,\ell_v\teq 1\in\Sc$, and so
  $(\ast)~\Ac\ent[\lian]\ell_z\teq 1+\ell_y+\ell_z+\ell_w$.
  Notice that for any variable $k$ of sort $\sseq$,
  we can apply $\rn{L-Valid}$, $\rn{L-Intro}$,
  and $\rn{S-Prop}$ to add to $\Ac$ either $\ell_k>0$ or
  $\ell_k=0$.
  Applying this to $y,z,w$,  we have that $\Ac\ent[\lian]\bot$ in each branch thanks
  to $(\ast)$, and  so \rn{A-Conf} applies and we get \unsat.
\end{example}

\subsection{Extended Calculus}
\begin{definition}
The calculus $\extcal$ is comprised of the derivation rules in
\Cref{fig:core-rules,fig:ext-rules}, with the addition of rule
\rn{R-Extract} from \Cref{fig:reduction-rules}.
\end{definition}

Our extended calculus combines array reasoning,
based on \cite{DBLP:conf/frocos/ChristH15}
and expressed by the rules in \Cref{fig:ext-rules},
with the core rules of \Cref{fig:core-rules}
and the \rn{R-Extract} rule.
Unlike in $\basiccal$, those rules do not reduce $\seqnth$ and $\sequpdate$.
Instead, they reason about those operators directly
and handle their combination with concatenation.
\rn{Nth-Concat} identifies the $i$th element of sequence $y$ with
the corresponding element selected from its normal form (see \Cref{lem:nf-unique}).
\rn{Update-Concat} operates similarly, applying $\sequpdate$ to
all the components. 
\rn{Update-Concat-Inv} operates similarly on the updated sequence
rather than on the original sequence.
\rn{Nth-Unit} captures the semantics of $\seqnth$ when applied to a
$\sequnit$ term. 
\rn{Update-Unit} is similar and distinguishes an update on an
out-of-bounds index (different from $0$) from an update within the bound.
\rn{Nth-Intro} is meant to ensure that \rn{Nth-Update} (explained below) and 
\rn{Nth-Unit} (explained above) are applicable whenever an $\sequpdate$ term exists
in the constraints.
\rn{Nth-Update} captures the read-over-write axioms of arrays,
adapted to consider their lengths
(see, e.g., \cite{DBLP:conf/frocos/ChristH15}).
It distinguishes three cases:
In the first, the update index is out of bounds.
In the second, it is not out of bounds, and the corresponding
$\seqnth$ term accesses the same index that was updated.
In the third case, the index used in the $\seqnth$ term is different
from the updated index.
\rn{Update-Bound} considers two cases:
either the update changes the sequence, or
the sequence remains the same.
Finally, \rn{Nth-Split} introduces a case split on the equality between
two sequence variables $x$ and $x'$ whenever they appear as arguments to $\seqnth$ with
equivalent second arguments. This is needed to ensure that we detect all cases where the arguments of two $\seqnth$ terms must be equal.

\begin{figure}[!ht]
\centering
\scalebox{.85}{
\begin{tabular}{cc}


\rn{Nth-Concat}
\(
\inferrule{
  x\teq\seqnth(y,i) \in \Sc
  \\
  \Sc \entnf^\ast y \teq \seqconiii{w_1}{\cdots}{w_n}
}{
  \Ac := \Ac, i < 0 \vee i \geq \ell_{y} 
  \ror \\\\
  \Ac := \Ac, 0 \leq i < \ell_{w_1} \\ \Sc := \Sc, x \teq \seqnth(w_1, i)
  \ror \ldots \ror \\\\
  \Ac := \Ac, \sum_{j=1}^{n-1}\ell_{w_j} \leq i < \sum_{j=1}^{n}\ell_{w_j} \\ \Sc := \Sc, x \teq \seqnth(w_n, i - \sum_{j=1}^{n-1}\ell_{w_j})
  \\\\
}
\)
\\[4ex]
\rn{Update-Concat}
\(
\inferrule{
  x\teq \sequpdate(y,i,v) \in \Sc
  \\
  \Sc \entnf^\ast y \teq \seqconiii{w_1}{\cdots}{w_n}
}{
  \Sc := \Sc, x \teq \seqconiii{z_1}{\cdots}{z_n},\\
  z_1\teq \sequpdate(w_1,i,v),\ldots,
  z_n\teq\sequpdate(w_n,i-\sum_{j=1}^{n-1}\ell_{w_j},v)
  \\\\
}
\)
\\[4ex]
\rn{Update-Concat-Inv}
\(
\inferrule{
  x\teq\sequpdate(y,i,v)\in\Sc
  \\
  \Sc \entnf^\ast x \teq \seqconiii{w_1}{\cdots}{w_n}
}{
  \Sc := \Sc, y \teq \seqconiii{z_1}{\cdots}{z_n}, \\\\
  w_1 \teq \sequpdate(z_1, i, v),\ldots, w_n \teq\sequpdate(z_n, i - \sum_{j=1}^{n-1} \ell_{w_j}, v)
  \\\\
}
\)
\\[4ex]
\rn{Nth-Unit}
\(
\inferrule{
  x\teq\seqnth(y, i) \in \Sc 
  \\
  \Sc \ent y\teq\sequnit(u)
}{
  \Ac:=\Ac,i<0 \vee i>0 
  \ror 
  \Ac:=\Ac,i \teq 0 \\
  \Sc:=\Sc,x \teq u
  \\\\
}
\)
\\[4ex]
\rn{Update-Unit}
\(
\inferrule{
  x\teq\sequpdate(y, i, v) \in \Sc 
  \\
  \Sc \ent y\teq\sequnit(u)
}{
  \Ac:=\Ac,i<0 \vee i>0 \\
  \Sc:=\Sc,x \teq \sequnit(u)
  \ror \\\\
  \Ac:=\Ac,i \teq 0  \\
  \Sc:=\Sc,x \teq \sequnit(v)
  \\\\
}
\)
\\[4ex]
\rn{Nth-Intro}
\(
\inferrule{
  s'\teq\sequpdate(s, i, t) \in \Sc
}{
  \Sc := \Sc, e \teq \seqnth(s, i),e'\teq\seqnth(s',i)
  \\\\
}
\)
\\[4ex]
 \rn{Nth-Update}
 \(
 \inferrule{
   \seqnth(x, j) \in \ter{\Sc} \\
   y\teq\sequpdate(z,i,v)\in\Sc \\ \Sc \ent x \teq y \mathrm{\ or\ } \Sc \ent x \teq z
 }{
   \Ac:=\Ac, j<0 \vee j\geq\ell_{x}
   \ror \\\\
   \Ac:=\Ac,i \teq j,0\leq j < \ell_{x} 
   \\
   \Sc:=\Sc,\seqnth(y,j)\teq v
   \ror \\\\
   \Ac:=\Ac,i\tneq j, 0\leq j < \ell_{x} 
   \\
   \Sc:=\Sc,\seqnth(y,j) \teq \seqnth(z,j)
   \\\\
 }
 \)
\\[4ex]
 \rn{Update-Bound}
 \(
 \inferrule{
   x\teq\sequpdate(y,i,v) \in \Sc
 }{
    \Ac:=\Ac, 0 \leq i < \ell_{y} \\
    \Sc:=\Sc, \seqnth(y,i)\tneq v
    \ror
    \Sc := \Sc, x \teq y
 }
 \)
 \\[4ex]
 \rn{Nth-Split}
 \(
 \inferrule{
   \seqnth(x,i),\seqnth(x',i') \in \ter{\Sc}
   \\
   i\teq i'\in\Ac
 }{
    \Sc:=\Sc, x\teq x'
    \ror
    \Sc := \Sc, x \tneq x'
 }
 \)
\end{tabular}}
\smallcap
\caption{Extended derivation rules. 
The rules use $z_1,\ldots,z_n$ to denote fresh sequence variables and
$e,e'$ to denote fresh element variables.
\label{fig:ext-rules}}
\vspace*{-1.8em}
\end{figure}

\subsection{Correctness}
\label{sec:correctness}
In this section we prove the following theorem:

\begin{restatable}{theorem}{correctness}
\label{thm:correctness}
Let $X\in\{\basiccal,\extcal\}$ and $\conf{\Sc_0,\Ac_0}$ 
be a configuration, and assume without loss of generality that $\Ac_0$ contains only arithmetic constraints that are not sequence constraints.
Let $T$ be a derivation tree obtained by applying the rules of $X$ with $\conf{\Sc_0,\Ac_0}$
as the initial configuration.
\begin{enumerate}
    \item\label{item:soundness} If $T$ is closed, then $\Sc_0\cup\Ac_0$ is $\sth$-unsatisfiable.
    \item\label{item:completeness} If $T$ contains a saturated configuration $\conf{\Sc,\Ac}$ w.r.t. $X$, then $\conf{\Sc,\Ac}$ is $\sth$-satisfiable, and so is $\conf{\Sc_0,\Ac_0}$.
\end{enumerate}
\end{restatable}


\noindent The theorem states that the calculi are correct in the following sense:
if a closed derivation tree is obtained for the constraints $\Sc_0 \cup \Ac_0$ 
then those constraints are unsatisfiable in $\sth$;
if a tree with a saturated leaf is obtained, then they are satisfiable.
It is possible, however, that neither kind of tree can be derived 
by the calculi, making them neither refutation-complete nor terminating.
This is not surprising since, as mentioned in the introduction, 
the decidability of even weaker theories is still unknown.

Proving the first claim in Theorem~\ref{thm:correctness} reduces 
to a local soundness argument for each of the rules. 
For the second claim,
we sketch below how to construct a satisfying model $\M$ 
from a saturated configuration for the case of $\extcal$. 
The case for $\basiccal$ is similar and simpler.

\paragraph{Model construction steps}
The full model construction and its correctness are described in
\begin{report}%
the appendix
\end{report}%
\begin{ijcar}%
a longer version of this paper~\cite{ShengEtAl-arXiv-22}
\end{ijcar}%
together with a proof of the theorem above.
Here is a summary of the steps needed for the model construction.

\begin{enumerate}
    \item\label{step:domains} Sorts: $\M(\selem)$ is interpreted as some
arbitrary countably infinite set.
$\M(\sseq)$ and $\M(\sint)$ are then determined by
the theory. 
\item\label{step:symbol} $\ssig$-symbols: $\sth$ enforces the interpretation of almost all $\ssig$-symbols, except
for
$\seqnth$ when the second input is out of bounds. 
We cover this case below.
\item\label{step:int-elem} Integer variables: based on the saturation of $\rn{A-Conf}$,
we know there is some $\lth$-model satisfying $\Ac$.
We set $\M$ to interpret integer variables according to this model.
\item\label{step:elem_var} Element variables: these are partitioned into their $\sceq{\Sc}$ equivalence classes.
Each class is assigned a distinct element
from $\M(\selem)$, which is possible since it is infinite.
\item Atomic sequence variables: these are assigned interpretations in several sub-steps:
\begin{enumerate}
    \item\label{step:length} length: we first use the assignments to variables $\ell_x$ to set the length of $\M(x)$, without assigning its actual value. 
    \item\label{step:unit} unit variables: for variables $x$ with $x\sceq{\Sc}\sequnit(z)$, we set $\M(x)$ to be $[\M(z)]$.
    \item\label{step:atomic-non-unit} non-unit variables: All other sequence variables are assigned values according to a {\em weak equivalence graph} we construct in a manner similar to \cite{DBLP:conf/frocos/ChristH15}. 
This construction takes into account constraints that involve $\sequpdate$ and $\seqnth$.
\end{enumerate}
\item\label{step:nonatomic_seq_var} Non-atomic sequence variables: these are first transformed to their unique normal form (see \Cref{lem:nf-unique}), consisting of concatenations of atomic variables. Then, the values assigned to these variables are concatenated.
\item\label{step:seqnth_terms} $\seqnth$-terms: for out-of-bounds indices in $\seqnth$-terms, we rely on $\sceq{\Sc}$ to make sure that the assignment is consistent. 
\end{enumerate}

\noindent
We conclude this section with an example of
the construction of $\M$.

\begin{example}
\label{ex:extcalmc}
Consider a signature in which $\selem$ is $\sint$, and 
a saturated configuration 
$\conf{\Sc^{\ast}, \Ac^{\ast}}$ 
w.r.t.
$\extcal$
that includes the following formulas:
    $y\teq \seqconii{y_1}{y_2}$,
    $x\teq \seqconii{x_1}{x_2}$,
    $y_2\teq x_2$,
    $y_1\teq\sequpdate(x_1, i, a)$,
    $\seqlen{y_1}=\seqlen{x_1}$,
    $\seqlen{y_2}=\seqlen{x_2}$,
    $\seqnth(y, i) \teq a$,
    $\seqnth(y_1, i)\teq a$.
    %
    Following the above construction, a satisfying interpretation $\M$ can be built as follows:
        \begin{description}
        \item[Step~\ref{step:domains}] Set both $\M(\sint)$ and $\M(\selem)$ to be the set of integer numbers. $\M(\sseq)$ is fixed by the theory.
        \item[Step~\ref{step:int-elem}, Step~\ref{step:elem_var}] First, find an arithmetic model,
        $\M(\ell_x)=\M(\ell_y)=4, \M(\ell_{y_1})=\M(\ell_{x_1})= 2, \M(\ell_{y_2})=\M(\ell_{x_2})=2, \M(i) = 0$.
        Further, set
        $\M(a)=0$.
        \item[Step~\ref{step:length}] Start assigning values to sequences. First, set the lengths of
        $\M(x)$ and $\M(y)$ to be $4$, and the lengths of
        $\M(x_1),\M(x_2),\M(y_1),\M(y_2)$ to be $2$.
        \item[Step~\ref{step:unit}] is skipped as there are no $\sequnit$ terms.
        \item[Step~\ref{step:atomic-non-unit}]
Set the $0$th element of $\M(y_1)$ to $0$
to satisfy $\seqnth(y_1,i)=a$ ($y_1$ is atomic, $y$ is not).
Assign fresh values to the remaining indices of atomic variables.
        The result can be, e.g.,
        $\M(y_1)=[0, 2], \M(x_1)=[1, 2], \M(y_2)=\M(x_2)=[3, 4]$. 
        \item[Step~\ref{step:nonatomic_seq_var}] Assign non-atomic sequence
          variables based on equivalent concatenations:
        $\M(y)=[0, 2, 3, 4], \M(x)=[1, 2, 3, 4]$.
        \item[Step~\ref{step:seqnth_terms}] No integer variable in the formula was assigned an out-of-bound value, and so
        the interpretation of $\seqnth$ on out-of-bounds cases
        is set arbitrarily. 
    \end{description}
\end{example}

\section{Implementation}
\label{sec:implementation}

We implemented our procedure for sequences
as an extension of a previous theory solver for strings~\cite{DBLP:conf/cav/LiangRTBD14, DBLP:conf/cav/ReynoldsWBBLT17}.
This solver is integrated in \cvc, 
and has been generalized to reason about both strings and sequences.
In this section, we describe how the rules of the calculus are implemented and the overall strategy for when they are applied.

Like most SMT solvers, \cvc is based on the \cdclt architecture~\cite{NieOT-JACM-06}
which combines several subsolvers, each specialized on a specific theory,
with a solver for propositional satisfiability (SAT).
Following that architecture, \cvc maintains an evolving set of formulas $\Fc$.
When $\Fc$ starts with quantifier-free formulas over the theory $\sth$,
the case targeted by this work,
the SAT solver searches for a satisfying assignment for $\Fc$, 
represented as the set $\Mc$ of literals it satisfies. 
If none exists, the problem is unsatisfiable at the propositional level
and hence $\sth$-unsatisfiable.
Otherwise, $\Mc$ is partitioned into the arithmetic constraints $\Ac$ and 
the sequence constraints $\Sc$ and checked for $\sth$-satisfiability
using the rules of the $\extcal$ calculus.
Many of those rules, including all those with multiple conclusions, 
are implemented by adding new formulas to $\Fc$ (following the
splitting-on-demand approach~\cite{BarNOT-LPAR-06}).
This causes the SAT solver to try to extend its assignment to those formulas,
which results in the addition of new literals to $\Mc$ 
(and thereby also to $\Ac$ and $\Sc$).

In this setting, the rules of the two calculi are implemented as follows.
The effect of rule \rn{A-Conf} is achieved by invoking \cvc's theory solver 
for linear integer arithmetic.
Rule \rn{S-Conf} is implemented by the congruence closure submodule of 
the theory solver for sequences.
Rules \rn{A-Prop} and \rn{S-Prop} are implemented by the standard mechanism for theory combination.
Note that each of these four rules may be applied \emph{eagerly}, that is, before constructing a complete satisfying assignment $\Mc$ for $\Fc$.

The remaining rules are implemented in the theory solver for 
sequences.
Each time $\Mc$ is checked for satisfiability, \cvc follows a strategy to determine which rule to apply next.
If none of the rules apply and the configuration is different from \unsat, 
then it is saturated, and the solver returns \sat.
The strategy for $\extcal$ prioritizes rules as follows.  
Only the first applicable rule is applied 
(and then control goes back to the SAT solver).
\begin{enumerate}
    \item (Add length constraints) For each sequence term in $\Sc$, apply \rn{L-Intro} or \rn{L-Valid}, if not already done. We apply \rn{L-Intro} for non-variables, and \rn{L-Valid} for variables.
    \item (Mark congruent terms) For each set of $\sequpdate$ (resp. $\seqnth$) terms that are congruent to one another in the current configuration, mark all but one term and ignore the marked terms in the subsequent steps.
    \item (Reduce $\seqextract$) For $\seqextract(y,i,j)$ in $\Sc$, apply \rn{R-Extract} if not already done.
    \item (Construct normal forms) Apply \rn{U-Eq} or \rn{C-Split}. We choose how to apply the latter rule based on constructing normal forms for equivalence classes in a bottom-up fashion, where the equivalence classes of $x$ and $y$ are considered before the equivalence class of $\seqconii{x}{y}$. We do this until we find an equivalence class such that $\Sc \entnf^\ast z \teq u_1$ and $\Sc \entnf^\ast z \teq u_2$ for distinct $u_1, u_2$.
    \item (Normal forms) Apply \rn{C-Eq} if two equivalence classes have the same normal form.
    \item (Extensionality) For each disequality in $\Sc$, apply \rn{Deq-Ext}, if not already done.
    \item\label{item:dist-update-nth} (Distribute $\sequpdate$ and $\seqnth$) For each term $\sequpdate(x,i,t)$ (resp. $\seqnth(x,j)$) such that the normal form of $x$ is a concatenation term, apply \rn{Update-Concat} and \rn{Update-Concat-Inv} (resp. \rn{Nth-Concat}) if not already done. Alternatively, if the normal form of the equivalence class of $x$ is a unit term, apply \rn{Update-Unit} (resp. \rn{Nth-Unit}).
    \item\label{item:array-reasoning} (Array reasoning on atomic sequences) Apply \rn{Nth-Intro} and \rn{Update-Bound} to $\sequpdate$ terms. For each $\sequpdate$ term, find the matching $\seqnth$ terms and apply $\rn{Nth-Update}$. Apply $\rn{Nth-Split}$ to pairs of $\seqnth$ terms with equivalent indices.
    \item (Theory combination) Apply \rn{S-A} for all arithmetic terms occurring in both $\Sc$ and $\Ac$.
\end{enumerate}

\noindent Whenever a rule is applied, the strategy will restart  
from the beginning in the next iteration. 
The strategy is designed to apply  with higher priority steps 
that are easy to compute and are likely to lead to conflicts. 
Some steps are ordered based on dependencies from other steps.
For instance, Steps 5 and 7 use normal forms, which are computed in Step 4.
The strategy for the $\basiccal$ calculus
is the same, except that Steps 7 and 8 are replaced by one 
that applies \rn{R-Update} and \rn{R-Nth} 
to all $\sequpdate$ and $\seqnth$ terms in $\Sc$.

We point out that the \rn{C-Split} rule may cause non-termination of 
the proof strategy described above in the presence of
\emph{cyclic} sequence constraints, for instance, constraints 
where sequence variables appear on both sides of an equality.
The solver uses methods for detecting some of these cycles, 
to restrict when \rn{C-Split} is applied.
In particular, when $\Sc \entnf^\ast x \teq \nf{(\seqconiii{\vec{u}}{s}{\vec{w}})}$,
$\Sc \entnf^\ast x \teq \nf{(\seqconiii{\vec{u}}{t}{\vec{v}})}$,
and $s$ occurs in $\vec{v}$, then \rn{C-Split} is not applied.
Instead, other heuristics are used, and in some cases the solver terminates with a response of ``unknown'' (see e.g.,~\cite{DBLP:conf/cav/LiangRTBD14} for details).
In addition to the version shown here, we also use another variation of the \rn{C-Split} rule 
where the normal forms are matched in reverse 
(starting from the last terms in the concatenations).
The implementation also uses fast entailment tests for length inequalities.
These tests may allow us to conclude which branch of \rn{C-Split}, if any, 
is feasible, without having to branch on cases explicitly.

Although not shown here,
the calculus can also accommodate certain \emph{extended} sequence constraints, that is, constraints using
a signature with additional functions.
For example, our implementation supports sequence containment, replacement, and reverse.
It also supports an extended variant of the $\sequpdate$ operator, in which the third argument is a sequence that
overrides the sequence being updated starting from the index given in the second argument.
Constraints involving these functions are handled by reduction rules, similar to those shown in 
\Cref{fig:reduction-rules}.
The implementation is further optimized by using context-dependent simplifications, which may eagerly infer when certain sequence terms can be simplified to constants based on the current set of assertions~\cite{DBLP:conf/cav/ReynoldsWBBLT17}.


\section{Evaluation}
\label{sec:evaluation}
We evaluate the performance of our approach,
as implemented in \cvc.
The evaluation investigates:
\begin{enumerate*}[label=(\roman*)]
    \item whether the use of sequences is a viable option for reasoning about vectors in programs,
    \item how our approach compares with other sequence solvers, and
    \item what is the performance impact of our array-style extended rules
\end{enumerate*}.
As a baseline, we use Version~4.8.14 of the \zzz SMT solver, which supports a theory of sequences without $\sequpdate$s.
For \cvc, we evaluate implementations of both the basic calculus (denoted \textbf{cvc5}) and the extended array-based calculus (denoted \textbf{cvc5-a}).
The benchmarks, solver configurations, and logs from our runs are available for download.\footnote{\url{http://dx.doi.org/10.5281/zenodo.6146565}}
We ran all experiments on a cluster equipped with Intel Xeon E5-2620 v4 CPUs.
We allocated one physical CPU core and 8GB of RAM for each solver-benchmark pair and used a time limit of 300 seconds.
We use the following two sets of benchmarks:

\newpar{Array Benchmarks (\textsc{Arrays})}
The first set of benchmarks is derived from the \verb|QF_AX| benchmarks in SMT-LIB~\cite{SMTLib2017}.
To generate these benchmarks, we
\begin{enumerate*}[label=(\roman*)]
\item replace declarations of arrays with declarations of sequences of uninterpreted sorts,
\item change the sort of index terms to integers, and
\item replace $\mathsf{store}$ with $\mathsf{update}$ and
 $\mathsf{select}$ with $\mathsf{nth}$
\end{enumerate*}.
The resulting benchmarks are quantifier-free and do not contain concatenations.
Note that the original and the derived benchmarks are not equisatisfiable, 
because sequences take into account out-of-bounds cases that do not occur in arrays.
For the \zzz runs, we add to the benchmarks a definition of $\sequpdate$ in terms of extraction and concatenation.
%

\newpar{Smart Contract Verification (\textsc{Diem})}
The second set of benchmarks consists of verification conditions generated by running the Move Prover~\cite{ZCQ+20} on smart contracts written for the Diem framework.
%
By default, the encoding does not use the sequence update operation, and so \zzz can be used directly.  
However, we also modified the Move Prover encoding to generate benchmarks that do use the update operator, and ran \cvc on them.
In addition to using the sequence theory, the benchmarks make heavy use of quantifiers and the SMT-LIB theory of datatypes.

\begin{figure}
    \centering
    \begin{subfigure}[b]{0.41\textwidth}
    \setlength{\tabcolsep}{4pt}
    \scalebox{0.9}{
    \begin{tabular}{@{}clrrrrrrrrrrrrrrrrrrrrr@{}}
\toprule
 & & \multicolumn{2}{c}{w/ update} & \\
\cmidrule(lr){3-4}
\textbf{Set} &  & \textbf{\cvc} & \textbf{\cvc-a} & \textbf{\ziii} \\
\midrule
\textsc{Arrays} & Slvd & 242& \textbf{390}& 170 \\
(551) & Time & 162& 303& 4329 \\
\midrule
\textsc{Diem} & Slvd & 542& \textbf{547}& 443 \\
(558) & Time & 518& 440& 639 \\
\bottomrule
\end{tabular}

    }
    \bigskip
    \caption{}
    \label{fig:table}
    \end{subfigure}
    \begin{subfigure}[b]{0.28\textwidth}
    \includegraphics[width=\linewidth]{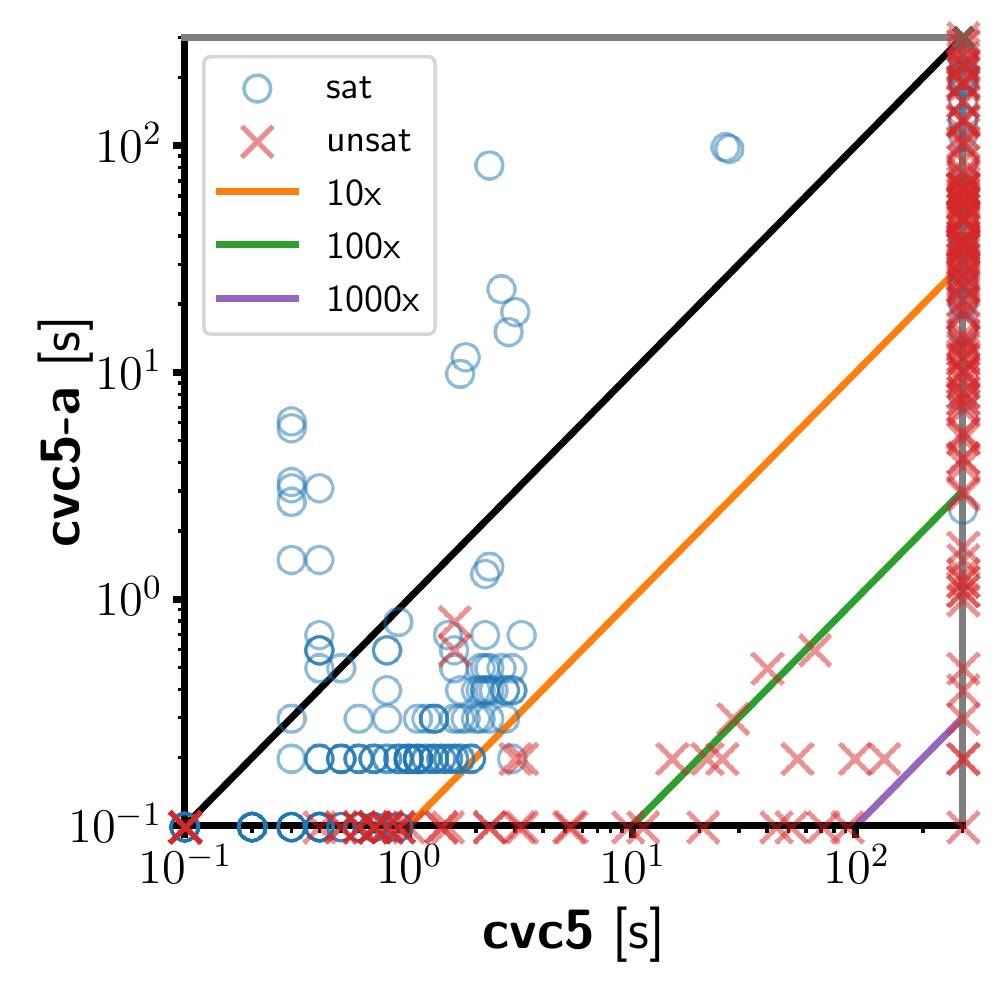}
    \caption{\textsc{Arrays}}
    \label{fig:qf-ax}
    \end{subfigure}
    \begin{subfigure}[b]{0.28\textwidth}
    \includegraphics[width=\linewidth]{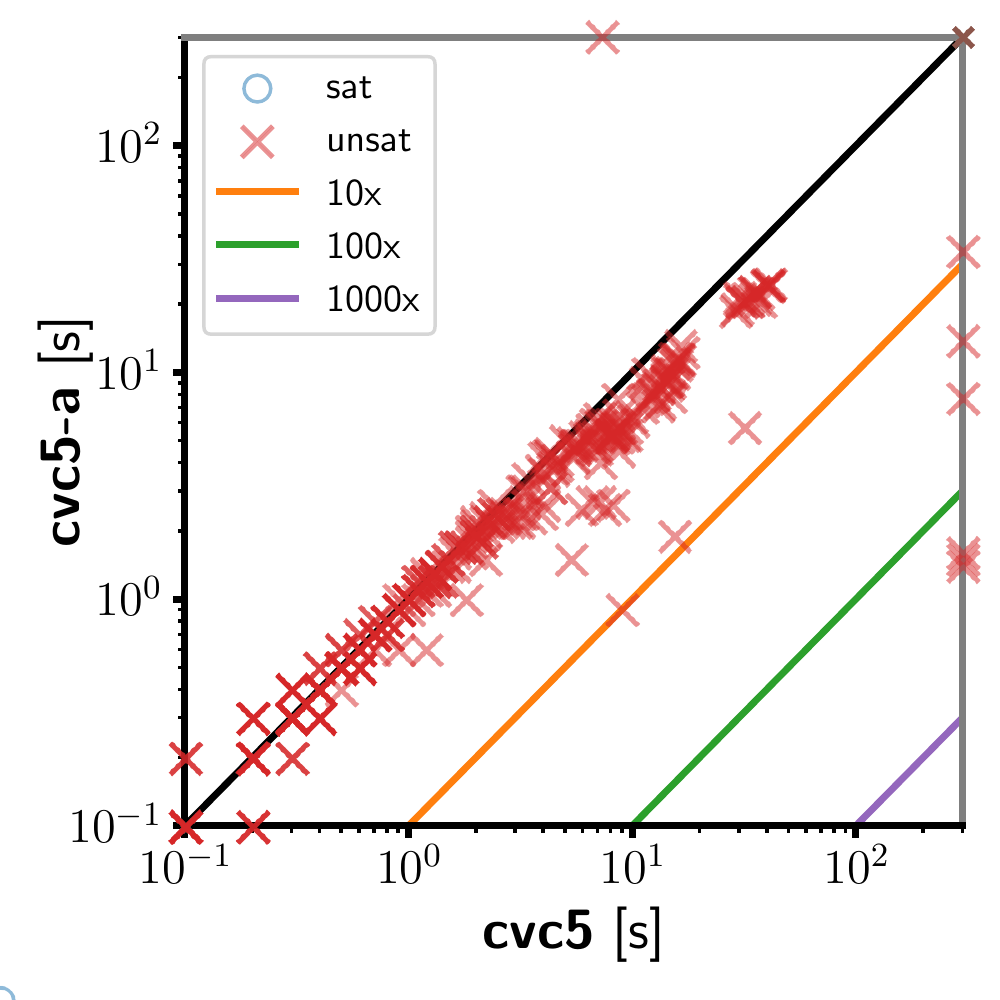}
      \caption{\textsc{Diem}}
    \label{fig:seq}
    \end{subfigure}

    \caption{
      \Cref{fig:table} lists the number of solved benchmarks and total time on commonly solved benchmarks.
      The scatter plots compare the base solver (\textbf{cvc5}) and the extended solver (\textbf{cvc5-a}) on \textsc{Array} (\Cref{fig:qf-ax}) and \textsc{Diem} (\Cref{fig:seq}) benchmarks.
    }
    \label{fig:results}
\end{figure}

\Cref{fig:table} summarizes the results 
in terms of number of
solved benchmarks and total time in seconds on commonly solved benchmarks.
The configuration that solves the largest number of benchmarks is the implementation of the extended calculus (\textbf{cvc5-a}).
This approach also successfully solves most of the \textsc{Diem} benchmarks,
which suggests that sequences are a promising option for encoding
vectors in programs.
The results further show that the sequences solver of \cvc significantly outperforms \zzz on both
the number of solved benchmarks and the solving time on commonly solved benchmarks.
%

\Cref{fig:qf-ax,fig:seq} show scatter plots comparing \textbf{cvc5} and \textbf{cvc5-a} on the two benchmark sets.  We can see a clear trend towards better performance when using the extended solver.  In particular, 
the table shows 
that in addition to solving the most benchmarks, \textbf{cvc5-a} is also fastest on the commonly solved instances from the \textsc{Diem} benchmark set.

For the \textsc{Arrays} set, we can see that some benchmarks are slower with the extended solver.  This is also reflected in the table, where \textbf{cvc5-a} is slower on the commonly solved instances.  This is not too surprising, as the extra machinery of the extended solver can sometimes slow down easy problems.  As problems get harder, however, the benefit of the extended solver becomes clear.  For example, if we drop \zzz and consider just the commonly solved instances between \textbf{cvc5} and \textbf{cvc5-a} (of which there are 242), \textbf{cvc5-a} is about $2.47\times$ faster (426 vs 1053 seconds).  Of course, further improving the performance of \textbf{cvc5-a} is something we plan to explore in future work.

\section{Conclusion}
\label{sec:conclusion}
We introduced calculi for checking satisfiability in the theory of sequences, which can be used to model 
the vector data type.
We described our implementation in \cvc and provided an evaluation, showing that the proposed
 theory is rich enough to naturally express verification conditions without
introducing quantifiers, and that our implementation is efficient. We believe that 
verification tools can benefit by changing their encoding of verification conditions
that involve vectors to use the proposed theory and implementation.

We plan to propose the incorporation of this theory in the SMT-LIB standard and
contribute our benchmarks to SMT-LIB.
As future research, we plan to integrate other approaches for array solving
into our basic solver. 
We also plan to study the politeness~\cite{JBLPAR,RRZ05} and decidability
of various fragments of the theory of sequences.

\bibliographystyle{abbrv}
\bibliography{main.bib}

\begin{report}
\newpage
\appendix
\section{Appendix}
In this appendix, each section presents a proof
that was omitted from the main text due to lack
of space.
For convenience, the lemma/theorem that is being proved is repeated at the beginning of each section.

Below, we make use of the following helper lemma.

\begin{lemma}
\label{lem:concat-term-len}
Let $\Sc$ be a set of sequence constraints
and $\Ac$ a set of arithmetic constraints.
Suppose $\conf{\Sc,\Ac}$ is saturated w.r.t.
$\rn{S-Prop}$,
$\rn{L-Intro}$
and $\rn{L-Valid}$.
If $\seqconiii{x_1}{\cdots}{x_n}\in\ter{\Sc}$ is a variable concatenation term of size $n$ not singular in $\Sc$, then
\begin{enumerate}
    \item $\Ac\ent[\lian] \Sigma_{k=1}^{n}\ell_{x_k} \ge 2$; and
    \item for each $m\in[1,n]$, $\Ac\ent[\lian] \Sigma_{k=1}^{n}\ell_{x_k} > \ell_{x_m}$.
\end{enumerate}

\end{lemma}
\begin{proof}
    Let $i,j\in[1,n], i\neq j$,
    such that
    $\Sc\notent x_i\teq \seqempty$ and
    $\Sc\notent x_j\teq\seqempty$.
    We know that $x_i,x_j\in\ter{S}$, so by saturation of 
    $\rn{L-Valid}$, we have
    $\Ac\ent\ell_{x_i}>0$ and $\Ac\ent\ell_{x_j}>0$. 
Also, by saturation of $\rn{L-Valid}$, $\rn{L-Intro}$, and $\rn{S-Prop}$, 
together with \Cref{assumption},
we know that $\Ac \ent[\lian] \ell_{x_k} \ge 0$ for $k\in[1,n]$.  It follows that $\Ac\ent[\lian]\Sigma_{k=1}^{n}\ell_{x_k}\ge 2$.  Furthermore,
    for each $m\in[1,n]$, 
    $\Ac\ent[\lian]\Sigma_{k=1}^{n}\ell_{x_k}>\ell_{x_m}$.  
\end{proof}

\subsection{Proof of \Cref{lem:flatas}}
\flatas*
\begin{proof}
Using standard transformations, $\varphi$ can be transformed into an equisatisfiable disjunction $\varphi'=\varphi_1\vee\ldots\vee\varphi_n$, where
for $i\in[1,n]$, $\varphi_i$ is a conjunction
$\varphi_i^{1},\ldots,\varphi_i^{n_i}$ of flat literals.
Each flat literal in $\varphi'$ is either a sequence constraint or an arithmetic constraint.
For each $i\in[1,n]$ we set 
$\Sc_i = \{\varphi_i^j\ |\ $such that $\varphi_i^j$ is a sequence constraint $\}$
and
$\Ac_i = \{\varphi_i^j\ |\ $such that $\varphi_i^j$ is an arithmetic constraint $\}$.
Hence, for every $\sth$-interpretation $\M$, 
$\M\models\varphi$ iff $\M\models \Sc_i\cup\Ac_i$ for some $i\in[1,n]$.
\end{proof}

\subsection{Proof of \Cref{lem:entnf-term}}
\entnfterm*
\begin{proof}

We show that the set of pairs $(x,s)$ for which $\Sc \entnf x \teq s$ is finite (it is then easy to see that the set of pairs $(y,t)$ for which $\Sc\entnf^{\ast} y\teq t$ is also finite).  Consider a tree whose root is obtained by \Cref{it:entnf-ref} or \Cref{it:entnf-con} of \Cref{def:entnf} where the children of a node are all possible results of applying \Cref{it:entnf-rec} of \Cref{def:entnf}.  Note that each node can have only finitely many children as there are only finitely many pairs $(y,t)$ where $y$ is a variable in $\Sc$ and $t$ is $\seqempty$ or a variable concatenation term in $\ter{S}$.  Below, we show that every path in the tree is finite, from which it follows that the tree is finite.  Since there are only finitely many such trees, it follows that the set of $(x,s)$ for which $\Sc \entnf x \teq s$ is finite.

    Define a partial order $\prec$ over terms of sort $\sseq$ in $\ter{\Sc}$, by
    $s\prec t$ iff $\Ac\ent[\lian]\ell_{x_s}<\ell_{x_t}$
    for some variables $x_s$ and $x_t$ such that
    $x_s\sceq{\Sc}s$ and $x_t\sceq{\Sc}t$.
    We show that
    $\prec$ is well defined.
    First, by \Cref{assumption}, such $x_s$ and $x_t$ exist.
    Next,
    if $x_s\sceq{\Sc}s$, $x_t\sceq{\Sc}t$,
    $x_s'\sceq{\Sc}s$ and $x_t'\sceq{\Sc}t$, then
    $\Ac\ent[\lian]\ell_{x_s}<\ell_{x_t}$  iff
    $\Ac\ent[\lian]\ell_{x_s'}<\ell_{x_t'}$,
    since $\Sc\ent x_s\teq x_s'$ and $\Sc\ent x_t\teq x_t'$
    and saturation of \rn{S-Prop}.
    Let $\prec^{\ast}$ be the (well-founded)
    Dershowitz-Manna
    multiset ordering
    induced by $\prec$ \cite{DBLP:journals/cacm/DershowitzM79}, 
    that is, the minimal transitive 
    relation that satisfies:
    $A\prec^{\ast} B$ whenever
    $A$ is obtained from $B$ by removing a single
    element $b$, and possibly adding any finite number
    of elements $a$ such that $a\prec b$ for every such
    new element.
    Now, for each term $t$ of sort $\sseq$ in $\ter{\Sc}$, let
    $m(t)$ be the multiset of the terms occurring in it
    (the multiplicity of each element in this multiset
    is the number of times it occurs in $t$).
    We prove the following claim, which establishes that every path in the tree mentioned above is finite: 
    if $\Sc \entnf x \teq \nf{(\seqconiii{\vec{w}}{y}{\vec{z}})}$,
    $\Sc \ent y \teq t$ and
    $t$ is $\seqempty$ or a variable concatenation term in $\Sc$ that is not singular in $\Sc$,
    then $m(\nf{(\seqconiii{\vec{w}}{y}{\vec{z}})})\succ^{\ast} m(\nf{(\seqconiii{\vec{w}}{t}{\vec{z}})})$.
    %
    
    
    %
    We consider two cases:
    $t$ is $\seqempty$ or $t$ is a variable concatenation term in $\ter{S}$
    not singular in $S$.
    In the first case,
    $\nf{(\seqconiii{\vec{w}}{t}{\vec{z}})}=
    \nf{(\seqconii{\vec{w}}{\vec{z}})}$. Note that the only role of $\downarrow$ here is to flatten nested concatenations.  Hence, $y$ is removed
    from the multiset of sub-terms and is not replaced
    by anything, so
    $m(\nf{(\seqconiii{\vec{w}}{y}{\vec{z}{}{}})})\succ^{\ast}m(\nf{(\seqconiii{\vec{w}}{t}{\vec{z}})})$.

    In the second case, $t$ is a variable concatenation term in $\Sc$ not singular in $\Sc$.  Let $t=\seqconiii{t_1}{\dots}{t_n}$.  Now, 
    $\nf{(\seqconiii{\vec{w}}{(\seqconiii{t_1}{\cdots}{t_n})}{\vec{z}})}$
 is a flat concatenation in which
    $y$ was removed, and $t_1,\ldots,t_n$ were added.
    To prove a decrease in $\prec^{\ast}$ from 
    $m(\nf{(\seqconiii{\vec{w}}{y}{\vec{z}})})$ to
    $m(\nf{(\seqconiii{\vec{w}}{t}{\vec{z}})})$,
    we show that for every $k\in[1,n]$,
    $t_k\prec y$, that is,
    $\Ac\ent[\lian]\ell_{t_k}<\ell_{y}$ (notice that $t_k$ and $y$ are variables).
    By \Cref{lem:concat-term-len}, we know that
    for each $k\in[1,n]$, 
    $\Ac\ent[\lian]\Sigma_{i=1}^{n}\ell_{t_i}>\ell_{t_k}$.  Since $\Sc\ent y\teq t$, we have $\Sc\ent\seqlen{y}=\seqlen{t}$. 
     By saturation of $\rn{L-Intro}$ and $\rn{S-Prop}$, and by \Cref{assumption},
    it follows that $\Ac\ent\ell_{y}=\Sigma_{i=1}^{n}\ell_{t_i}$.  Thus, 
$\Ac\ent[\lian]\ell_{y}>\ell_{t_k}$.

\end{proof}

\subsection{Proof of \Cref{lem:nf-unique}}
\nfunique*

\begin{proof}

Let $e$ be an equivalence class of $\sceq{\Sc}$ whose terms are of sort $\sseq$, and let $x\in e$ (we know every equivalence class has at least one variable by \Cref{lem:equiv_class_nonempty}). We show existence and uniqueness of a normal form for $e$.

\medskip
\noindent
{\bf Existence: }
We show that for some $\vec{s}$, $\Sc\entnf^{\ast}x\teq \vec{s}$.
Consider the tree construction of \Cref{lem:entnf-term}.  If we follow some path in the tree, we will reach a node $\Sc \entnf x\teq \vec{t}$ for which no children can be derived (i.e., \Cref{it:entnf-rec} of \Cref{def:entnf} doesn't apply).  It is not hard to see that $\vec{t} = (t_1,\dots,t_n)$, where $n\ge 0$ and each $t_i$ is a variable for $i\in[1,n]$.

Let $\vec{s}$ have the same length as $\vec{t}$ and let $s_i:=\cf(t_i)$ for $i\in[1,n]$.
We prove that $\Sc\entnf^{\ast}x\teq \vec{s}$.
We know that $\Sc\entnf x\teq \vec{t}$.  If $n=0$, then trivially, $\Sc\entnf^{\ast}x\teq \vec{s}$.  Otherwise,
for each $i\in[1,n]$, we further know that $\Sc\ent t_i\teq s_i$ and $s_i = \cf([s_i])$, so it only remains to show that
$s_i$ is atomic in $\Sc$.
Assume it is not.
Then either $\Sc \ent s_i\teq \seqempty$ or
there exists a variable concatenation term $\vec{u} \in \ter{\Sc}$
not singular in $\Sc$ such that
$\Sc\ent s_i\teq \vec{u}$.  In either case, this would imply that \Cref{it:entnf-rec} of \Cref{def:entnf} is applicable to $\Sc \entnf x\teq \vec{t}$, contradicting our assumption.

\medskip
\noindent
{\bf Uniqueness: }
By the existence argument above, there exists some $\vec{s}$ such that
$\Sc\entnf^{\ast}x\teq \vec{s}$.
Now, suppose $\Sc\entnf^{\ast}x\teq \vec{s'}$, and assume
that $\vec{s}\neq \vec{s'}$.
Then there must be $\vec{w},\vec{z},\vec{z'},y,y'$, each containing only variables that are atomic representatives,
such that $s=\nf{(\seqconiii{\vec{w}}{y}{\vec{z}})}$ and
$s'=\nf{(\seqconiii{\vec{w}}{y'}{\vec{z'}})}$, with $y \neq y'$.
By saturation w.r.t. $\rn{C-Split}$, there are three possibilities:
$\Ac\ent\ell_{y}>\ell_{y'}$ and $y\teq \seqconii{y'}{k}\in\Sc$ for some $k$;
$\Ac\ent\ell_{y}<\ell_{y'}$ and $y'\teq \seqconii{y}{k}\in\Sc$ for some $k$; or
$\Ac\ent\ell_{y} = \ell_{y'}$ and $y\teq y'\in\Sc$.
In the first case, notice that since $\Sc\ent y\teq \seqconii{y'}{k}$, it follows that $\Sc\ent \seqlen{y} \teq \seqlen{\seqconii{y'}{k}}$.  Also, by saturation of $\rn{L-Intro}$, $\Sc\ent\seqlen{\seqconii{y'}{k}} \teq \seqlen{y'} + \seqlen{k}$.  So, $\Sc \ent \seqlen{y} \teq \seqlen{y'} + \seqlen{k}$.  And, by saturation of $\rn{S-Prop}$ and  \Cref{assumption} also $\Ac\ent\ell_{y} \teq \ell_{y'} + \ell_{k}$.  It follows that $\Sc \notent k \teq \seqempty$; otherwise, we would have $\Ac\ent\ell_{k}\teq 0$ by $\rn{L-Intro}$ and $\rn{S-Prop}$, which together with $\Ac\ent\ell_{y} > \ell_{y'}$ contradicts saturation of $\rn{A-Conf}$.
Also, $y'$ is atomic and hence, $\Sc\notent y'\teq\seqempty$.
Thus, $\Sc\ent y \teq \seqconii{y'}{k}$ but $\seqconii{y'}{k}$ is not singular in $\Sc$
as $\Sc\notent k\teq\seqempty$ and $\Sc\notent y'\teq\seqempty$.
In particular, this means that $y$ is not atomic, which contradicts our assumption, so
the first case is impossible. 
The second case is analogous to the first.
In the third case, we have
$y\teq y'\in\Sc$, but we know that $y$ and $y'$ are both equivalence class representatives, so it must be that $y = y'$, which contradicts our assumption that $y\neq y'$.

\end{proof}

\subsection{Proof of \Cref{thm:correctness}, \Cref{item:soundness}}

\correctness*

The proof of \Cref{item:soundness} is routine, and amounts to a local soundness check of each of the derivation rules. 
The only non-trivial case is $\rn{C-Split}$, which relies on the following lemma:

\begin{lemma}
\label{lem:entnf-sound}
Let $\Sc$ be a set of sequence constraints
and $\Ac$ a set of arithmetic constraints.
Suppose $\conf{\Sc,\Ac}$ is saturated w.r.t.
$\rn{S-Prop}$,
$\rn{L-Intro}$
and $\rn{L-Valid}$.
If $\Sc\entnf^{\ast}x\teq s$ then $\Sc \ent[\sth] x\teq s$.
\end{lemma}

\begin{proof}
Suppose $\Sc \entnf^\ast x \teq \vec{s}$ where $\vec{s} = (s_1, \ldots, s_n)$.  Then, for some $\vec{s'} = (s'_1, \ldots, s'_n)$,
$\Sc \entnf x \teq \vec{s'}$ and
$\Sc \ent s_i \teq s'_i$ for $i \in [1,n]$.  To show that     $\Sc\ent[\sth]x\teq \vec{s}$, it thus suffices to show that $\Sc \ent[\sth] x \teq \vec{s'}$.
    As shown in \Cref{lem:entnf-term}, there is a finite number $k$ of derivation steps
    in $\entnf$ that yield $x\teq \vec{s'}$.
    We prove the claim by induction on $k$.
    For $k=1$, we either apply \Cref{it:entnf-ref} or \Cref{it:entnf-con}.
    In the former case we get a trivial identity $\Sc\entnf x\teq x$,
    and clearly
    $\Sc\ent[\sth]x\teq x$.
    In the latter case, we get an identity $\Sc\entnf x\teq t$ such that
    $\Sc\ent x\teq t$, and so in particular
    $\Sc\ent[\sth]x\teq t$.
    Suppose $k>1$, and that
    $x\teq \vec{s'}$ was obtained using
    \Cref{it:entnf-rec} of \Cref{def:entnf}.
    Hence, $\vec{s'}$ has the form
    $\nf{(\seqconiii{\vec{w}}{t}{\vec{z}})}$, where
    $\Sc\entnf x\teq\nf{(\seqconiii{\vec{w}}{y}{\vec{z}})}$ with
    a shorter derivation,
    and $\Sc\ent y\teq t$.
    By the induction hypothesis,
    $\Sc\ent[\sth]x\teq\nf{(\seqconiii{\vec{w}}{y}{\vec{z}})}$.
    Clearly, 
    $\Sc\ent[\sth]y\teq t$.
    Hence,
    $\Sc\ent[\sth]x\teq \vec{s'}$.
\end{proof}

Using \Cref{lem:entnf-sound},
since $\Sc\entnf^{\ast}x\teq s$ implies that
$x\teq s$ follows from $\Sc$ in $\sth$, we have that
every model of the assumptions of $\rn{C-Split}$
assigns the same values to 
$\nf{(\seqconiii{\vec{w}}{y}{\vec{z}})}$ and
$\nf{(\seqconiii{\vec{w}}{y'}{\vec{z'}})}$.
Assuming that the first two branches of the conclusion
do not hold in some model, we indeed get that
the interpretations of $y$ and $y'$ must be
identical, as they are the sub-sequence of the same sequence,
that begins at the same index and whose length is the same.\qed

\subsection{Proof of \Cref{thm:correctness}, \Cref{item:completeness}}

\correctness*
    
Consider a saturated configuration $\conf{\Sc, \Ac}$.
We define an interpretation $\M$.
Unless stated otherwise, by \define{equivalence class} we mean an equivalence class
w.r.t. $\sceq{\Sc}$ and may drop the subscript when referring to it
(e.g., writing $[x]$ for the equivalence class of $x$).
We also sometimes drop ``{\em in $\Sc$}'' from
``{\em atomic in $\Sc$}'' when $\Sc$ is clear from context.
We say that an equivalence class $e$ is \define{atomic} if its terms are of sort $\sseq$ and all variables in it are atomic.
Note that if $e$ contains any variable that is atomic in $\Sc$, then all variables in $e$ are atomic in $\Sc$.

The following definitions completely define the model $\M$.
Concretely, 
\Cref{def:domains} defines the domains of $\M$;
\Cref{def:mc-fun} defines the function symbols of $\ssig$.
\Cref{def:mc-int-elem} defines the interpretations of variables of sort $\sint$ and $\selem$;
\Cref{def:mc-length} assigns lengths to the interpretations of variables of sort $\sseq$ that are atomic in $\Sc$ (it does not set their values, but only their lengths);
\Cref{def:mc-unit} defines the interpretations of variables of sort $\sseq$ that are equivalent to $\sequnit$ terms;
\Cref{def:mc-modelcon} defines the interpretations of variables of sort $\sseq$ that are atomic in $\Sc$ but are not equivalent to $\sequnit$ terms;
\Cref{def:mc-non-atom} defines the interpretations of variables of sort $\sseq$ that are not atomic in $\Sc$;
Finally, \Cref{def:mc-nth} defines the interpretation of the $\seqnth$ symbol in $\M$ (note that only a part of its interpretation is fixed by the theory).

The definitions are shown to be well-defined in \Cref{lem:wd}. 

We start by assigning domains to the sorts of $\sth$.

\begin{definition}[Model construction: domains]
\label{def:domains}
\begin{enumerate}
    \item $\M(\sint)=\mathbb{Z}$, the set of integers.
    \item $\M(\selem)$ is some countably infinite set.
    \item $\M(\sseq)$ is the set of finite sequences whose elements are taken from $\M(\selem)$.
\end{enumerate}
\end{definition}
The domains for $\sint$ and $\sseq$ were chosen to be consistent with the requirements of
$\sth$. The exact identity of $\selem$ is not important, and so we set
its elements to be arbitrary. In contrast, its cardinality
is important, and is set to be infinite.

\begin{definition}[Model construction: function symbols]
\label{def:mc-fun}
    The symbols shown in \Cref{fig:sig} are interpreted in such a way that $\M$ is both a $\lth$-interpretation and a $\sth$-interpretation, as described in \Cref{fig:sig} and in Sections \ref{sec:preliminaries} and \ref{sec:theory}.
\end{definition}

\noindent
Now we can assign values to variables of sort $\sint$ and $\selem$. 

\begin{definition}[Model construction: $\sint$ and $\selem$ variables]
\label{def:mc-int-elem}
\begin{enumerate}
    \item\label{item:mc-arith} Let $\arithmodel$ be a $\lth$-interpretation with $\arithmodel(\sint)=\mathbb{Z}$
    that satisfies $\Ac$. Then 
    $\M(x):=\arithmodel(x)$ for every variable $x$ of sort $\sint$.
    \item\label{item:mc-elem} Let $a_1,a_2,\ldots$ be an enumeration of $\M(\selem)$, and let
    $e_1,\ldots,e_n$ be an enumeration of the equivalence classes of
    $\sceq{\Sc}$ whose variables have sort $\selem$. Then, for every $i\in [1,n]$ and every variable $x\in e_i$, $\M(x):=a_i$.
    \end{enumerate}
\end{definition}

\begin{lemma}
\label{lem:eq_len}
If $y\teq \sequpdate(x,i,z)\in\Sc$, then $\M(\ell_{x})=\M(\ell_{y})$.
\end{lemma}

\begin{proof}

By congruence, we have that $\Sc\ent \seqlen{y} \teq \seqlen{\sequpdate(x,i,z)}$, and by saturation of $\rn{L-Intro}$ and the definition of $\downarrow$, we know that $\Sc\ent \seqlen{\sequpdate(x,i,z)}\teq \seqlen{x}$.  Thus, $\Sc \ent \seqlen{x} \teq \seqlen{y}$. By \Cref{assumption}, $\Sc\ent\ell_x\teq\ell_y$, so by saturation of $\rn{S-Prop}$, we have $\Ac\ent\ell_x\teq\ell_y$.  It follows from \Cref{def:mc-int-elem} that $\arithmodel(\ell_x)=\arithmodel(\ell_y)$, and thus, $\M(\ell_x)\teq\M(\ell_y)$.
\end{proof}

We now turn to assigning values to variables of sort $\sseq$.
We will first associate a sequence value with every sequence equivalence class (we write $\M(e)$ for the value associated with equivalence class $e$).
Then, for each equivalence class $e$ and for each $x\in e$, we define $\M(x) := \M(e)$.  We start by defining lengths.

\begin{definition}[Model construction: length of atomic variables]
\label{def:mc-length}
For each variable $x$ in an atomic equivalence class $e$ of $\sceq{\Sc}$,
we constrain $\M(e)$ to be a sequence of length $\M(\ell_x)$.
(the elements of $\M(e)$ are not defined yet).
\end{definition}
Note that because every atomic variable is in an atomic equivalence class, \Cref{def:mc-length} assigns a length to every atomic variable.
In what follows, we denote the length assigned to $\M(e)$ in \Cref{def:mc-length} by $\ell_{e}$.

Next, we define model values for atomic variables.
We start with equivalence classes that contain $\sequnit$-terms.

\begin{definition}
An equivalence class $e$ of $\sceq{\Sc}$ is called a {\em unit equivalence class} if
$\sequnit(x)\in e$ for some $x$.
\end{definition}

\begin{definition}[Model construction: atomic variables in unit equivalence classes]
\label{def:mc-unit}
For every atomic variable $x$ such that $x \sceq{\Sc} \sequnit(y)$ for some $y$,
we set $\M([x])$ to be a sequence of length $1$
whose only element is set to $\M(y)$.
\end{definition}

Next we turn to atomic equivalence classes that are not unit. 
We begin by defining a graph in order to keep track of constraints
that originate from the $\sequpdate$ symbol.
This definition is an adaptation of the weak equivalence graph from 
\cite{DBLP:conf/frocos/ChristH15} to the context of sequences, where
operations are richer.
In particular, we take into account not only whether
a constraint of the form
$y\teq\sequpdate(x,i,a)$ appears in $\Sc$, but also
whether $x$ and $y$ are atomic, and whether 
the interpretation already given to $i$ in \Cref{def:mc-int-elem}
is within the range determined by the length assigned to
$x$ and $y$ in \Cref{def:mc-length}.

\begin{definition}[Weak equivalence graph]
\label{def:mc-weak-equiv}
Define a graph $G=(V,E,\delta)$ as follows.
$V$ is the set of atomic equivalence classes. $E\subseteq V\times V$ is a set of unordered edges, and $\delta: E\rightarrow\powerset{\mathbb{N}}$ is a labeling function on edges, such that $(e_1,e_2)\in E$ and $k\in\delta((e_1,e_2))$ iff
there are $x\in e_1$ and $y\in e_2$ such that
    $y\teq\sequpdate(x, i, z) \in \Sc$ or
    $x\teq\sequpdate(y, i, z) \in \Sc$, where
    $\M(i)=k$ and
    $0\leq k < \M(\ell_{e_1}) = \M(\ell_{e_2})$.
\end{definition}

\begin{definition}[Weak equivalence]
\label{def:weak_equiv_rel}
Given a weak equivalence graph $G=(V,E,\delta)$, for each $i\in\mathbb{N}$, we define
a binary relation $\sim_i$ over $V$
as follows: $e_1 \sim_i e_2$ iff there exists a path $p$ between
$e_1$ and $e_2$ in $G$,
such that for each edge $d$ in $p$, $\delta(d)\setminus \{i\} \neq \emptyset$.
\end{definition}

It is routine to verify that:

\begin{lemma}
\label{lem:simi-eq}
For each $i\in\mathbb{N}$, $\sim_i$ is an equivalence relation over the atomic equivalence classes of $\sceq{\Sc}$.
\end{lemma}

\noindent
Notice that even though $\sim_i$ is defined using $\sequpdate$-terms,
every atomic equivalence class $e$ satisfies $e \sim_i e$, even if
it has no $\sequpdate$-terms.

Let $\ell=\max_{e \in \sceq{\Sc}} \ell_e$ be the maximal sequence length assigned in
\Cref{def:mc-length}.
Let $a_1^{1},a_1^{2},\ldots, a_1^{\ell},a_2^{1},\ldots,a_2^{\ell},\ldots$ be an enumeration of the elements in $\M(\selem)$ that were not assigned to any variable of $\selem$ sort.
For each $i\in[0,\ell)$, let $E_1^{i},E_2^{i},\ldots,E_{n_{i}}^{i}$ be an enumeration of the equivalence classes of $\sim_i$.

By \Cref{lem:simi-eq}, we have:
\begin{lemma}
\label{lem:eqcsimi}
For every atomic equivalence class $e$ of $\sceq{\Sc}$,
and for each $i\in[0,\ell_{e})$,
there exists some $j$ such that
$e\in E_{j}^{i}$.
\end{lemma}

\begin{definition}[Model construction: atomic variables in non-unit equivalence classes]
\label{def:mc-modelcon}
Let $e$ be an atomic equivalence class of $\sceq{\Sc}$ that is not a unit equivalence class.
We set the $i$th element of $\M(e)$ for every $i\in[0,\ell_e)$
as follows.
By \Cref{lem:eqcsimi}, there exists some $j$
such that $e\in E_{j}^{i}$.
\begin{enumerate}
    \item\label{item:update-prop} If there are $e'\in E_j^i$, $s\in e'$, $x$, and $y$ such that
    $x\sceq{\Sc}\seqnth(s,y)$ 
    and $\M(y)=i$,
    we set the $i$th element of $\M(e)$ to be $\M(x)$.
    \item\label{item:update-fresh} Otherwise, the $i$th element of $\M(e)$ is set to $a_j^{i}$.
\end{enumerate}
\end{definition}

Next, we set the interpretation of non-atomic variables of sort $\sseq$.

\begin{definition}[Model construction: non-atomic sequence variables]
\label{def:mc-non-atom}
Let $x$ be a non-atomic variable of sort $\sseq$, and let $\vec{y}$ be its normal form.
$\M([x])$ is set to be the concatenation of the interpretations of the variables in $\vec{y}$ ($\M([x])$ is the empty sequence if $\vec{y}$ is of size 0).
\end{definition}

Finally, we define the interpretation of $\seqnth$.
$\M(\seqnth)$ is a function from $\M(\sseq)\times\M(\sint)$ to $\M(\selem)$.
Given $a\in\M(\sseq)$ and an integer $i$, if
$i$ is non-negative and smaller than the length of $a$,
the value of $\M(\seqnth)(a,i)$ is fixed by the theory.
The following definition assigns a value to $\M(\seqnth)(a,i)$ when $i$ is out of bounds.

\begin{definition}[Model construction: $\seqnth$ terms]
\label{def:mc-nth}
For every element $a\in\M(\sseq)$ with length $n$ and for every $i\in\mathbb{N}$
such that $i<0$ or $i \ge n$,
if there are $s$, $x$, and $y$ such that $x\sceq{\Sc}\seqnth(s,y)$, $\M(s)=a$,
    and $\M(y)=i$,
    $\M(\seqnth)(a,i)$ is set to $\M(x)$.
Otherwise, it is set arbitrarily.
\end{definition}

This concludes the construction of $\M$.

\begin{restatable}{lemma}{lemwd}
\label{lem:wd}
\Cref{def:domains,def:mc-fun,def:mc-int-elem,def:mc-length,def:mc-unit,def:mc-modelcon,def:mc-non-atom,def:mc-nth} are well-defined.
\end{restatable}
This lemma is proved in \Cref{sec:proofoflemwd}, below.

By construction, $\M$ is a $\sth$-interpretation, and by \Cref{lem:wd}
it is well-defined. It is left to show that it satisfies $\Ac\cup\Sc$.
The arithmetic constraints are satisfied, by
the fact that $\arithmodel\models\Ac$
(see \Cref{def:mc-fun,def:mc-int-elem}):
\begin{lemma}
\label{lem:arith}
$\M$ satisfies $\Ac$.
\end{lemma}

Showing that $\M\models \Sc$ is more involved.
Roughly speaking, each possible shape of a sequence constraint is considered separately,
and for each, $\M$ is proved to satisfy constraints of that shape from $\Sc$.
The cases that include $\seqnth$ and $\sequpdate$ terms heavily rely
on the construction of the weak equivalence graph.
Constraints that include concatenation are handled by reasoning about
$\entnf$.

\begin{restatable}{lemma}{msatsc}
\label{lem:MsatSc}
$\M$ satisfies $\Sc$.
\end{restatable}
This lemma will be proved in \Cref{sec:proofofmsatsc}.

We conclude this section by reviewing the earlier example for
the construction of $\M$ with a few more details.

\begin{example}
\label{ex:extcalmc}
Consider a signature in which $\selem$ is $\sint$, and
a saturated configuration 
$\conf{\Sc^{\ast}, \Ac^{\ast}}$ 
w.r.t.
$\extcal$
that includes the following formulas:
    $y\teq \seqconii{y_1}{y_2}$,
    $x\teq \seqconii{x_1}{x_2}$,
    $y_2\teq x_2$,
    $y_1\teq\sequpdate(x_1, i, a)$,
    $\seqlen{y_1}=\seqlen{x_1}$,
    $\seqlen{y_2}=\seqlen{x_2}$,
    $\seqnth(y, i) \teq a$,
    $\seqnth(y_1, i)\teq a$.
    %
    Following the above construction, a satisfying interpretation $\M$ can be built as follows:
    \begin{description}
        \item[\Cref{def:domains}] Set both $\M(\sint)$ and $\M(\selem)$ to be the set of integer numbers. $\M(\sseq)$ is fixed by the theory.
        \item[\Cref{def:mc-int-elem}] First, find an arithmetic model,
        $\M(\ell_x)=\M(\ell_y)=4, \M(\ell_{y_1})=\M(\ell_{x_1})= 2, \M(\ell_{y_2})=\M(\ell_{x_2})=2, \M(i) = 0$.
        Further, set
        $\M(a)=0$.
        \item[\Cref{def:mc-length}] Start assigning values to sequences. First, set the lengths of
        $\M(x)$ and $\M(y)$ to be $4$, and the lengths of
        $\M(x_1),\M(x_2),\M(y_1),\M(y_2)$ to be $2$.
        \item[\Cref{def:mc-modelcon}] \Cref{def:mc-unit} is skipped as there are no $\sequnit$ terms.
        Next, according to \Cref{item:update-prop}, the $0$th element of $\M(y_1)$ is set to $0$ ($y_1$ is atomic, $y$ is not.).
        According to \Cref{item:update-fresh}, assign fresh values to the remaining indices of atomic variables.
        The result can be, e.g.,
        $\M(y_1)=[0, 2], \M(x_1)=[1, 2], \M(y_2)=\M(x_2)=[3, 4]$. 
        \item[\Cref{def:mc-non-atom}] Assign non-atomic sequence variables
          based on equivalent concatenations:
        $\M(y)=[0, 2, 3, 4], \M(x)=[1, 2, 3, 4]$.
        \item[\Cref{def:mc-nth}] No integer variable in the formula was assigned an out-of-bound value, and so
        the interpretation of $\seqnth$ on out-of-bounds cases
        is set arbitrarily. 
    \end{description}
\end{example}

\subsection{Proof of \Cref{lem:wd}}
\label{sec:proofoflemwd}

\lemwd*

Note that \Cref{def:domains,def:mc-fun} are trivially well-defined.
We now go through the remaining definitions.

\begin{lemma}
\Cref{def:mc-int-elem}
is well-defined.
\end{lemma}

\begin{proof}
To show that 
\Cref{item:mc-arith} is well-defined, it suffices to note that by saturation of $\rn{A-Conf}$, $\Ac$ is $\lth$-satisfiable. To show that \Cref{item:mc-elem} is well-defined, it suffices to establish an infinite
enumeration $a_1,a_2,\ldots$ of $\M(\selem)$, which is guaranteed to exist
due to \Cref{def:domains}.
\end{proof}

\begin{lemma}
\label{lem:mc-length}
\Cref{def:mc-length} is well-defined.
\end{lemma}

\begin{proof}
Let $x,y\in e$. We show that $\M(\ell_{x})=\M(\ell_{y})$.
Since $x\sceq{\Sc}y$, we have
$\Sc\ent \seqlen{x}\teq\seqlen{y}$.
By \Cref{assumption}, we also have
$\Sc\ent \ell_{x}\teq \ell_{y}$.
By saturation of $\rn{S-Prop}$, we have
$\ell_{x}\teq\ell_{y}\in\Ac$, and hence,
by \Cref{def:mc-int-elem},
$\M(\ell_{x})=\M(\ell_{y})$.
\end{proof}

\begin{lemma}
\label{lem:unit-wd}
\Cref{def:mc-unit} is well-defined.
\end{lemma}

\begin{proof}
We first show that $\ell_{[x]}$ is 1.  To see this, note first that by saturation of $\rn{L-Intro}$ (and \Cref{fig:nf}), we have $\Sc\ent \seqlen{\sequnit(y)} \teq 1$.  We also have $\ell_x \teq \seqlen{x} \in \Sc$ by \Cref{assumption}.  It follows that $\Sc \ent \ell_x \teq 1$, so $\ell_x \teq 1\in\Ac$ by saturation of $\rn{S-Prop}$.  Thus, we must have $\arithmodel(\ell_x)=1$ in \Cref{def:mc-int-elem}, and thus $\M(\ell_x)=1$.  By \Cref{def:mc-length}, we then have that the length of $\M([x])$ must be 1.

Next, suppose $\sequnit(y)$, $\sequnit(z)\in [x]$.
We prove $\M(y)=\M(z)$.
Since $\Sc\ent \sequnit(y)\teq \sequnit(z)$, 
we must have $y\teq z\in\Sc$, by saturation of $\rn{U-Eq}$.
Hence, $y\sceq{\Sc}z$ and so by \Cref{def:mc-int-elem},
we have $\M(y)=\M(z)$.
\end{proof}

\begin{lemma}
\Cref{def:mc-modelcon} is well-defined.
\end{lemma}
\begin{proof}
Let $e$ be as in \Cref{def:mc-modelcon}.
Suppose there are $e',e''\in E_j^i$, with $s'\in e'$, $s''\in e''$, and $x', y', x'', y''$ such that 
$x'\sceq{\Sc}\seqnth(s',y')$ and 
$x''\sceq{\Sc}\seqnth(s'',y'')$ and
$\M(y')=\M(y'')=i$, with $i\in[0,\ell_e)$.
We prove that $\M(x')=\M(x'')$.

We first show that $y'\teq y''\in\Sc$.  Clearly $y',y''\in\ter{\Sc}$.  Also, $y',y''\in\ter{\Ac}$ since $y'\teq y'\in\Ac$ and $y''\teq y''\in\Ac$ by saturation of $\rn{S-Prop}$.  It follows by saturation of $\rn{S-A}$ that either $y'\teq y''\in\Ac$ or $y'\tneq y''\in\Ac$.
Since $\M(y')=\M(y'')$, the latter cannot hold, and so the former holds.
Then, by saturation of $\rn{A-Prop}$, we have $y'\teq y''\in\Sc$ as well.

Now, notice that $e'\sim_i e''$.
Hence, there exist $e_1,\ldots, e_n$ such that
$e_1=e'$, $e_n=e''$ and for $k\in[1,n)$, $e_k$ and $e_{k+1}$
are connected by an edge $d_k$ in $G$ where $\delta(d_k)\setminus \{i\} \neq \emptyset$.
For every such $k$, we have that
$s_k\!\teq\!\sequpdate(s_k', y_{k}, z_{k})\in\Sc$ or
$s_k'\!\teq\!\sequpdate(s_{k}, y_{k}, z_{k})\in\Sc$
for some $s_k\in e_k$, $s_k'\in e_{k+1}$, integer variable $y_k$, and $\selem$-variable $z_k$.
By \Cref{lem:eq_len}, $\M(\ell_{s_k})=\M(\ell_{s'_k})$.  And by
\Cref{def:mc-length}, $\ell_{e_k} = \M(\ell_{s_k})$ and $\ell_{e_{k+1}} = \M(\ell_{s'_k})$.  It follows that
$\ell_{e'}=\ell_{e_1}=\ldots=\ell_{e_n}=\ell_{e''}$.  By a similar argument, because $e \sim_i e'$, we have $\ell_{e'}=\ell_e$.
We also have:
$(\ast)$~$\M(y_{k})\neq i$ and
$(\ast\ast)$~$\M(y_k)\in[0,\ell_e)$.
Define $s_0'$ to be an alias for $s'$ and then notice that for $k\in[1,n)$, $s_k\sceq{\Sc} s_{k-1}'$ because $s_k, s_{k-1}'$ are both in $e_k$.

We prove by induction that for $k\in[0,n)$, $\seqnth(s',y')\sceq{\Sc}\seqnth(s'_k,y')$. For the base case, we simply note that $\seqnth(s',y')$ and $\seqnth(s'_0,y')$ are identical and $\seqnth(s',y')\in\ter{\Sc}$.
For the induction step, suppose that $\seqnth(s',y')\sceq{\Sc}\seqnth(s'_k,y')$, where $k\in[0,n-1)$.
This implies $\seqnth(s'_k,y')\in\ter{\Sc}$, and we also know $s'_k \sceq{\Sc}s_{k+1}$.  Recalling that $s_{k+1}\teq\sequpdate(s_{k+1}', y_{k+1}, z_{k+1})\in\Sc$ or $s_{k+1}'\teq\sequpdate(s_{k+1}, y_{k+1}, z_{k+1})\in\Sc$, we see that the premises of $\rn{Nth-Update}$ are satisfied.
By saturation of $\rn{Nth-Update}$, then, there are three possibilities.
\begin{enumerate}
\item In the first case, $\Ac\ent y'<0 \vee y'\ge \ell_{s'_k}$.  
We know from \Cref{def:mc-int-elem} that $\M(y')=\arithmodel(y')$ and $\M(\ell_{s'_k})=\arithmodel(\ell_{s'_k})$, so we must have $\M(y')<0$ or $\M(y') \ge \M(\ell_{s'_k})$.  But $\M(y')=i$ and $i\in[0,\ell_e)$, and we know that $\M(\ell_{s'_k})=\ell_{e_{k+1}}=\ell_e$, so this case is not possible.
\item In the second case, $\Ac\ent y' \teq y_{k+1}$.  This is also not possible because we know that $\arithmodel(y')=\M(y')=i\neq\M(y_{k+1})=\arithmodel(y_{k+1})$.
\item We are thus left with the third option,
in which $\seqnth(s_{k+1},y') \sceq{\Sc} \seqnth(s'_{k+1},y')$. But we know that $s'_k\sceq{\Sc} s_{k+1}$, so we also have $\seqnth(s'_k,y') \sceq{\Sc} \seqnth(s'_{k+1},y')$.  Then, by the induction hypothesis, $\seqnth(s',y') \sceq{\Sc} \seqnth(s'_{k+1},y')$, which completes the induction proof.
\end{enumerate}
Letting $k=n-1$, we obtain $\seqnth(s',y')\sceq{\Sc}\seqnth(s'_{n-1},y')$.  But $s'_{n-1}\in e_n$ and $e_n = e''$, so $s'_{n-1}\sceq{\Sc} s''$ and $\seqnth(s',y')\sceq{\Sc}\seqnth(s'',y')$.  Finally, since we showed above that $y'\teq y''\in\Sc$, we have $\seqnth(s',y')\sceq{\Sc}\seqnth(s'',y'')$, and thus 
$x'\sceq{\Sc}x''$, which means that $\M(x')=\M(x'')$
by \Cref{def:mc-int-elem}.
\end{proof}

\begin{lemma}
\Cref{def:mc-non-atom} is well-defined.
\end{lemma}

\begin{proof}
$\vec{y}$ exists and is unique by \Cref{lem:nf-unique}.
If $\vec{y}$ is a variable or a variable-concatenation term, then uniqueness guarantees well-definedness.
Further, each variable that occurs in $\vec{y}$ is atomic by \Cref{lem:nf-unique}, and hence its value
in $\M$ was already defined in \Cref{def:mc-unit,def:mc-modelcon}.
\end{proof}

We prove that \Cref{def:mc-nth} is well-defined, but first we prove some helper lemmas.
Recall that we write $\Cc\ent[\lian]\varphi$ to denote that every model of $\lth$ satisfying $\Cc$ also satisfies $\varphi$.
Intuitively, if $\varphi$ can be derived from $\Cc$ using arithmetic reasoning, then $\Cc\ent[\lian]\varphi$.

\begin{lemma}
\label{lem:entnf-len}
If $\Sc\entnf x\teq\vec{z}$ and $\vec{z}$ is of size $n$, then $\Ac\ent[\lian] \ell_{x}=\Sigma_{i=1}^n\ell_{z_i}$.
\end{lemma}
\begin{proof}
The proof is by structural induction using \Cref{def:entnf}.  
For \Cref{it:entnf-ref}, clearly $\Ac\ent\ell_x\teq \ell_x$.
For \Cref{it:entnf-con}, 
we have $\Sc\entnf x\teq t$ for some 
variable concatenation term $t=\seqconiii{t_1}{\cdots}{t_n}$
such that
$\Sc\ent x\teq t$.
Therefore, $\Sc\ent \seqlen{x}\teq\seqlen{t}$. 
By saturation of $\rn{L-Intro}$,
$\Sc\ent\seqlen{x}\teq\Sigma_{i=1}^{n}\seqlen{t_i}$,
and using
\Cref{assumption}, we get
$\Sc\ent\ell_x\teq\Sigma_{i=1}^{n}\ell_{t_i}$.
By saturation of $\rn{S-Prop}$,
$\Ac\ent\ell_x\teq\Sigma_{i=1}^{n}\ell_{t_i}$.

Now suppose that $\Sc\entnf x\teq \nf{(\seqconiii{\vec{w}}{y}{\vec{z}})}$ and $\Sc\ent y\teq t$, where $t$ is $\seqempty$ or a variable concatenation term in $\ter{S}$.  Let $\vec{w}=(w_1,\dots,w_m)$ and $\vec{z}=(z_1,\dots,z_n)$, with $m,n\ge0$.  By the induction hypothesis, we have that $\Ac\ent[\lian]\ell_{x}\teq\Sigma_{i=1}^m \ell_{w_i} + \ell_{y} + \Sigma_{i=1}^n \ell_{z_i}$.  We consider two cases. (1) $t=\seqempty$: in this case, $\Sc\ent\seqlen{y}\teq\seqlen{\seqempty}$.  Also, $\Sc\ent\seqlen{\seqempty}=0$ by saturation of $\rn{L-Intro}$, so $\Sc\ent\seqlen{y}=0$, and thus $\Ac\ent\ell_{y}=0$ by saturation of $\rn{S-Prop}$ and \Cref{assumption}.  It follows that $\Ac\ent[\lian]\ell_{x}\teq\Sigma_{i=1}^m\ell_{w_i} + \Sigma_{i=1}^n\ell_{z_i}$. (2) $t$ is a variable concatenation term in $\ter{\Sc}$: let $t=\seqconiii{t_1}{\dots}{t_k}$.  We have $\Sc\ent\seqlen{y}\teq\seqlen{t}$.  Also, $\Sc\ent\seqlen{t}=\Sigma_{i=1}^{k}\seqlen{t_i}$ by saturation of $\rn{L-Intro}$.  By saturation of $\rn{S-Prop}$ and \Cref{assumption} (and assuming wlog that $\Sigma_{i=1}^k\ell_{t_i}$ is the result of flattening $\Sigma_{i=1}^{k}\seqlen{t_i}$), it follows that $\Ac\ent\ell_{y}=\Sigma_{i=1}^{k}\ell_{t_i}$.  Thus, $\Ac\ent[\lian]\ell_{x} \teq \Sigma_{i=1}^m\ell_{w_i} + \Sigma_{i=1}^k\ell_{t_i} + \Sigma_{i=1}^n \ell_{z_i}$.
\end{proof}

\begin{lemma}
\label{lem:seqlen}
For every sequence variable $x\in\Sc$, if $\ell$ is the length of $\M(x)$, then $\M(\ell_x)=\ell$.
\end{lemma}
\begin{proof}
If $x$ is atomic, then $\ell=\M(\ell_x)$ by \Cref{def:mc-length}.  Suppose that $x$ is non-atomic.  Let $\vec{y}$ be the normal form of $x$ where $\vec{y}$ is of size $n$.  Each element of $\vec{y}$ is atomic, so for $i\in[1,n]$, the length of $\M(y_i)$ is $\M(\ell_{y_i})$ by \Cref{def:mc-length}.  Then, $\ell=\Sigma_{i=1}^n \M(\ell_{y_i})$ by \Cref{def:mc-non-atom}.  Let $\vec{z}$ of length $n$ be such that $\Sc\entnf x\teq\vec{z}$ and $\Sc\ent y_i\teq z_i$ for $i\in[1,n]$, which exists by \Cref{def:entnf}.  By \Cref{lem:entnf-len}, we have that $\Ac\ent[\lian] \ell_{x}\teq\Sigma_{i=1}^n\ell_{z_i}$.
For each $i\in[1,n]$, we know that because $\Sc\ent y_i\teq z_i$,
$\Sc\ent \seqlen{y_i}\teq\seqlen{z_i}$, and so
by \Cref{assumption},
$\Sc\ent \ell_{y_i}\teq\ell_{z_i}$.
Therefore,
$\Ac\ent\ell_{y_i}\teq\ell_{z_i}$ by saturation of $\rn{S-Prop}$.  
Then,  we can conclude that $\Ac\ent[\lian] \ell_x\teq\Sigma_{i=1}^n\ell_{y_i}$.  Finally, we have $\M(\ell_x)=\arithmodel(\ell_x)=\Sigma_{i=1}^n\arithmodel(\ell_{y_i})=\Sigma_{i=1}^n\M(\ell_{y_i})=\ell$.

\end{proof}

\begin{lemma}
\label{lem:nthinrange-atomic}
Suppose $k\sceq{\Sc}\seqnth(x,y)$.
Let $i=\M(y)$, $e=[x]$, and let $\ell_e$ be the length of $\M(x)$.  If $e$ is atomic and $i\in[0,\ell_e)$, then the $i$th element of $\M(x)$ is $\M(k)$.
\end{lemma}

\begin{proof}
We consider the following cases:
\begin{enumerate}
    \item $e$ is a unit equivalence class:
    then, for some $z$, $\sequnit(z)\in e$.  By \Cref{def:mc-unit},
    we have that
    $i=0$ and $\M(x)$ contains the single element $\M(z)$.  We show that $\M(z)=\M(k)$.
    Since $\seqnth(x,y)\in\ter{\Sc}$ and
    $\Sc\ent x\teq \sequnit(z)$, by saturation of
    $\rn{Nth-Unit}$, there are two cases.
    In the first case, $\Ac\ent y<0\vee y>0$,
    which, by \Cref{def:mc-int-elem}, is not possible since $\M(y)=i=0$.
    So, we are left with the second case, in which $\seqnth(x,y)\teq z\in\Sc$.
    It follows that $z\sceq{\Sc} k$, so 
    by \Cref{def:mc-int-elem}, $\M(z)=\M(k)$.

    \item $e$ is atomic but not a unit equivalence class:
    let $j$ be such that $e\in E_j^{i}$.
    Then, by \Cref{item:update-prop} of \Cref{def:mc-modelcon}, the $i$th element of $\M(e)$ must be $\M(k)$, and thus the $i$th element of $\M(x)$ is $\M(k)$.
\end{enumerate}
\end{proof}

\begin{lemma}
\label{lem:nthinrange}
Suppose $k\sceq{\Sc}\seqnth(x,y)$.
Let $i=\M(y)$, $e=[x]$, and let $\ell_e$ be the length of $\M(x)$.  If $i\in[0,\ell_e)$, then the $i$th element of $\M(x)$ is $\M(k)$.
\end{lemma}
\begin{proof}
  If $e$ is atomic, then we have the result by \Cref{lem:nthinrange-atomic}.
  Suppose $e$ is not atomic.  Let $\vec{y}$ be the normal form of $x$.
    \begin{enumerate}
    \item If $\vec{y}$ is empty, then by \Cref{def:mc-non-atom}, $\M(x)$ is the empty sequence, so $i\in[0,\ell_e)$ is always false, and the statement holds vacuously.

    \item Suppose $\vec{y}$ has a single element, $y_1$.  By \Cref{def:entnf}, it's clear that $\Sc\entnf^{\ast}y_1\teq y_1$.  So, by saturation of $\rn{C-Eq}$, we must have $x\teq y_1\in\Sc$.  But $y_1$ is atomic, so this contradicts the assumption that $e$ is not atomic.

    \item Otherwise, $\vec{y}=\seqconiii{y_1}{\cdots}{y_n}$, with $n\ge 2$.
    Recall that $\seqnth(x,y)\in\ter{\Sc}$. Thus, by saturation of $\rn{Nth-Concat}$, one of its $n+1$ conclusions is applicable.
    In the first case, we must have $\Ac\ent y<0\vee y\geq\ell_{x}$.  But we also know that $\M(y)=i$ is non-negative and is smaller than the length assigned to $M(x)$, which leads to a contradiction using \Cref{lem:seqlen}.
    For the other cases, we have, for some $k\in[1,n]$,
    (1) $\Ac\ent \Sigma_{j=1}^{k-1}\ell_{y_j}\leq y < \Sigma_{j=1}^{k}\ell_{y_j}$ and
    (2) $\Sc\ent \seqnth(x,y)\teq\seqnth(y_k,y-\Sigma_{j=1}^{k-1}\ell_{y_j})$.
    By \Cref{def:mc-int-elem}, this means that $\Sigma_{j=1}^{k-1}\M(\ell_{y_j})\leq\M(y)<\Sigma_{j=1}^k\M(\ell_{y_j})$.  Now, by \Cref{def:mc-non-atom}, $\M(x)=\seqconiii{\M(y_1)}{\cdots}{\M(y_n)}$, and by \Cref{lem:seqlen}, the length of $\M(y_j)$ is $\M(\ell_{y_j})$ for $j\in[1,n]$.  Let $i'=\M(y)-\Sigma_{j=1}^{k-1}\M(\ell_{y_j})$.  Clearly, $i'\in[0,\M(\ell_{y_k}))$, and element $\M(y)$ of $\M(x)$ is the same as element $i'$ of $\M(y_k)$.  Now, revisiting (2), let $\alpha$ be the term $y-\Sigma_{j=1}^{k-1}\ell_{y_j}$, and let $\hat{\alpha}$ be the variable introduced for $\alpha$ when flattening the term $\seqnth(y_k,\alpha)$.  We have $k\sceq{\Sc}\seqnth(x,y)$, so $k\sceq{\Sc}\seqnth(y_k,\hat{\alpha})$ by (2).  Let $i''$ be $\M(\hat{\alpha})$, and recall that $y_k$ is atomic and that the length of $\M(y_k)$ is $\M(\ell_{y_k})$.  By \Cref{lem:nthinrange-atomic}, we have that if $i''\in[0,\M(\ell_{y_k}))$, then the $i''$th element of $\M(y_k)$ is $\M(k)$.  It remains to show that $i'=i''$.  To see this, note that $\Sc\ent\hat{\alpha}\teq\alpha$.  So, by saturation of $\rn{S-Prop}$, we have $\Ac\ent \hat{\alpha}\teq\alpha$. Then, by \Cref{def:mc-int-elem}, $i''=\M(\hat{\alpha})=\M(y)-\Sigma_{j=1}^{k-1}\M(\ell_{y_j})=i'$.
    \end{enumerate}    
\end{proof}

\begin{lemma}
\label{lem:disequal}
For all $x, y$, with $x\tneq y\in \Sc$, we have $\M(x) \neq \M(y)$. 
\end{lemma}

\begin{proof}
By $\rn{Deq-Ext}$ we have two cases. In the first,
 $\Ac\ent \ell_x\neq\ell_y$.  So, by \Cref{def:mc-int-elem}, $\M(\ell_x)\neq\M(\ell_y)$.  Thus, by \Cref{lem:seqlen}, we have that the length of $\M(x)$ is different from the length of $\M(y)$, so $\M(x) \neq \M(y)$.

In the second case, we have
(1) $\Ac\ent \ell_x \teq\ell_y \wedge 0\leq i <\ell_x$ and
(2) $w_1\teq\seqnth(x,i),w_2\teq\seqnth(y,i),w_1\tneq w_2\in\Sc$, for some $i,w_1,w_2$.
By (2) and saturation of $\rn{S-Conf}$, we know that $w_1 \scneq{\Sc} w_2$, so
by \Cref{def:mc-int-elem}, $\M(w_1)\neq\M(w_2)$.
Let $n=\M(i)$ and let $\ell$ be the length of $\M(x)$.
By \Cref{def:mc-int-elem}, we have $\M(\ell_x)=\M(\ell_y)$ and $0\leq n < \M(\ell_x)$.  So, by \Cref{lem:seqlen}, we have that
the lengths of $\M(x)$ and $\M(y)$ are both equal to $\ell$, and $n\in[0,\ell)$.
Looking again at (2), we can apply \Cref{lem:nthinrange} twice to get that
the $n$th element of $\M(x)$ is $\M(w_1)$ and
the $n$th element of $\M(y)$ is $\M(w_2)$.
We can then conclude that $\M(x)\neq \M(y)$,
as we know that $\M(w_1)\neq\M(w_2)$.
\end{proof}

We can now show that \Cref{def:mc-nth} is well-defined.

\begin{lemma}
\label{lem:nth-wd}
\Cref{def:mc-nth} is well-defined.
\end{lemma}

\begin{proof}
Suppose there are
$x,x',s,s',y,y'$ such that
$x\sceq{\Sc}\seqnth(s,y)$,
$x'\sceq{\Sc}\seqnth(s',y')$,
$\M(s) = \M(s') = a$, and
$\M(y) = \M(y') = i$.
We prove $\M(x)=\M(x')$.
Since $y,y'\in\ter{\Sc}$, and they have sort $\sint$,
by saturation of $\rn{S-Prop}$, we have that
$y\teq y,y'\teq y'\in\Ac$, and so
$y,y'\in\ter{\Sc}\cap\ter{\Ac}$. 
By saturation of $\rn{S-A}$, either $y\tneq y'\in \Ac$ or $y \teq y'\in \Ac$.
The first case is impossible since $\M(y)=\M(y')$.
In the second case, we have $y\teq y'\in\Ac$, and so $y\teq y'\in\Sc$ by saturation of \rn{A-Prop}.
Now, by saturation of $\rn{Nth-Split}$,
there are two options: either $s\tneq s'\in\Sc$ or $s \teq s'\in\Sc$.
The first is impossible by
\Cref{lem:disequal}, as
$\M(s)=\M(s')$.
On the other hand, if $s\teq s'\in\Sc$,
then since 
$x\sceq{\Sc}\seqnth(s,y)$ and 
$x'\sceq{\Sc}\seqnth(s',y')$,
we also have 
$x\sceq{\Sc}x'$.
Thus, by \Cref{def:mc-int-elem}, we have $\M(x)=\M(x')$.

\end{proof}

\subsection{Proof of \Cref{lem:MsatSc}}
\label{sec:proofofmsatsc}

\msatsc*

We start with helper lemmas.

\begin{lemma}
\label{lem:atomic-nf}
$x$ is atomic in $\Sc$ iff $\Sc\entnf^{\ast}x\teq y$ for some atomic representative $y$.
\end{lemma}
\begin{proof}
    $\Rightarrow$: Let $y$ be the representative of $[x]$.  By \Cref{def:atomic}, and because $x\sceq{\Sc} y$ and $x$ is atomic, $y$ must also be atomic.  We have $\Sc\entnf x\teq x$ by \Cref{def:entnf}, and thus, also by \Cref{def:entnf}, we have $\Sc\entnf^{\ast}x\teq y$.
    
    $\Leftarrow$: Since $y$ is atomic and $x\sceq{\Sc}y$, it is easy to see by \Cref{def:atomic} that $x$ is also atomic.
\end{proof}

\begin{lemma}
\label{lem:model-eq-nf}
If $\Sc\entnf^{\ast}x\teq\vec{y}$, then $\M(x) = \seqconiii{\M(y_1)}{\dots}{\M(y_n)}$, where $\vec{y}$ has size $n\ge 0$.
\end{lemma}
\begin{proof}
If $x$ is not atomic, then the result is immediate by \Cref{def:mc-non-atom}.  If $x$ is atomic, then by \Cref{lem:atomic-nf} and uniqueness of normal forms (\Cref{lem:nf-unique}), $n=1$.  Since $x\sceq{\Sc} y_1$, and models are assigned by equivalence class in \Cref{def:mc-unit,def:mc-modelcon}, it follows that $\M(x)=\M(y_1)$.
\end{proof}

\begin{lemma}
\label{lem:seqemptylen}
$\Sc\ent x\teq\seqempty$ iff $\M(x)$ has length 0.
\end{lemma}
\begin{proof}
$\Rightarrow$: If $\Sc\ent x\teq\seqempty$, then $\Sc\ent\seqlen{x}\teq\seqlen{\seqempty}$.  By saturation of \rn{L-Intro}, $\Sc\ent\seqlen{\seqempty}=0$.  So, by \Cref{assumption}, $\Sc\ent \ell_x\teq 0$, and so, by saturation of \rn{S-Prop}, we have $\ell_x\teq 0\in\Ac$.  Thus, by \Cref{def:mc-int-elem}, $\M(\ell_x)=0$, and so the length of $x$ is 0 by \Cref{lem:seqlen}.
$\Leftarrow$: If $\M(x)$ has length 0, then by \Cref{lem:seqlen}, $\M(\ell_x)=0$.  By saturation of \rn{L-Valid}, either $x\teq\seqempty\in\Sc$ or $\Ac\ent\ell_x>0$.  But the latter is impossible by \Cref{def:mc-int-elem}.
\end{proof}

\begin{lemma}
\label{lem:update-align}
If $x \teq\sequpdate(y, i, v) \in \Sc$, $\M(i)\in[0,\M(l_y))$, and $\Sc\entnf^\ast y \teq \seqconiii{w_1}{\cdots}{w_n}$, then:
\begin{enumerate}
    \item\label{it:update-align-concat} $x\teq \seqconiii{z_1}{\cdots}{z_n}\in\Sc$ for some atomic $z_1,\ldots,z_n$;
    \item\label{it:update-align-k} there exists some $k\in[1,n]$, such that $\sum_{j=1}^{k-1} \M(\ell_{w_j}) \leq \M(i) < \sum_{j=1}^k \M(\ell_{w_j})$ and $z_k=\sequpdate(w_k, \alpha_k, v)\in\Sc$, where $\Ac\ent \alpha_k\teq i-\sum_{j=1}^{k-1}\ell_{w_j}$; and
    \item\label{it:update-align-j} for all $j\in[1,n]$, $j\neq k$, $z_j\teq w_j\in \Sc$.
\end{enumerate}
\end{lemma}
\begin{proof}
By saturation w.r.t. \rn{Update-Concat},
there is $x \teq \seqconiii{z_1}{\cdots}{z_n}\in \Sc$,
for some $z_1,\ldots, z_n$, such that 
$z_m \teq \sequpdate(w_m, \alpha_m, v)\in\Sc$ for $m\in[1,n]$,
where $\alpha_m$ is the variable introduced for the term $i-\sum_{j=1}^{m-1}\ell_{w_j}$ when flattening the term $\sequpdate(w_m, i-\sum_{j=1}^{m-1}\ell_{w_j}, v)$.
Then, $z_m\teq\sequpdate(w_m, \alpha_m, v)\in \Sc$ and $\alpha_m \sceq{\Sc} i-\sum_{l=1}^{m-1}\ell_{w_l}$.
By saturation of \rn{S-Prop}, we also have $\Ac\ent \alpha_m \teq i-\sum_{j=1}^{m-1}\ell_{w_j}$.

Next we prove that $z_1,..., z_n$ are atomic.
Suppose $z_m$ is not atomic for some $m\in[1,n]$.  Note that we cannot have $\Sc\ent z_m\teq\seqempty$: $w_m$ is atomic, so by \Cref{lem:seqemptylen}, $\M(w_m)$ has a nonzero length; then, by \Cref{lem:eq_len}, $\M(z_m)$ has nonzero length, so $\Sc\notent z_m\teq\seqempty$ by \Cref{lem:seqemptylen}.  Let $\vec{u}$ be the normal form of $z_m$. $\vec{u}$ cannot be empty because then, by \Cref{lem:model-eq-nf}, the length of $\M(z_m)$ would be 0, so by \Cref{lem:seqemptylen}, we would have $z_m\teq\seqempty$.  $\vec{u}$ cannot be of size 1 as then $z_m$ would be atomic by \Cref{lem:atomic-nf}.  Thus, $\vec{u}$ is of size at least 2.  Now, by saturation of \rn{Update-Concat-Inv} applied to $z_m\teq\sequpdate(w_m, \alpha_m, v)$, we have $w_m\teq\vec{z'}\in\Sc$, $u_1\teq\sequpdate(z'_1,\alpha_m,v)\in\Sc$, and $u_2\teq \sequpdate(z'_2,\alpha',v)\in\Sc$, for some $\alpha'$ where $\Sc\ent\alpha'=\alpha_m-\ell_{u_1}$.  By \Cref{lem:eq_len}, we have $\M(\ell_{u_1})=\M(\ell_{z'_1})$ and $\M(\ell_{u_2})\teq\M(\ell_{z'_2})$.  But $u_1$ and $u_2$ are atomic, so by \Cref{lem:seqemptylen}, their lengths cannot be zero.  By \Cref{lem:seqlen}, then, the lengths of $z'_1$ and $z'_2$ are also nonzero.  So $\Sc\notent z_1\teq\seqempty$ and $\Sc\notent z_2\teq\seqempty$ by \Cref{lem:seqemptylen}.  Thus, $\vec{z}$ is not singular, which means that $w_m$ is not atomic, which is a contradiction.

Now, consider the following $n+2$ constraints: $i < 0$, $\sum_{j=1}^{k-1} \ell_{w_j} \leq i < \sum_{j=1}^k \ell_{w_j}$ for $k\in[1,n]$, and $i\geq\sum_{j=1}^n \ell_{w_j}$.  Exactly one of these holds in $\M$, since it interprets arithmetic symbols in the usual way by \Cref{def:mc-fun}.

Suppose $\M\ent i < 0$ or $\M\ent i\geq\sum_{j=1}^n \ell_{w_j}$.
We know that $\M(i)\in[0,\M(\ell_y))$, so this case is impossible by \Cref{lem:model-eq-nf,lem:seqlen}.

Now, suppose that $\M\ent \sum_{j=1}^{k-1} \ell_{w_j} \leq i < \sum_{j=1}^k \ell_{w_j}$ for some $k \in [1,n]$.  Clearly, \Cref{it:update-align-k} holds for this $k$.
We know that \rn{Update-Bound} is saturated w.r.t. $z_m \teq \sequpdate(w_m, \alpha_m, v)\in\Sc$ for $m\in[1,n]$.  Recall also that $\Ac\ent \alpha_m \teq i-\sum_{j=1}^{m-1}\ell_{w_j}$.  It is easy to see that the first branch is inconsistent with $\M$ whenever $m\neq k$.  Thus, we have $z_m\teq w_m\in\Sc$ for $m\in[1,n],m\neq k$.
\end{proof}

\begin{lemma}
\label{lem:unitatomic}
If $x\sceq{\Sc}\sequnit(y)$ then $x$ is atomic.
\end{lemma}
\begin{proof}
Note that $\Sc\ent\seqlen{\sequnit(y)}\teq 1$ by saturation of $\rn{L-Intro}$, and so we also have ($\ast$) $\Sc\ent\seqlen{x}\teq 1$.
Assume $x$ is not atomic. There are two cases.  In the first case, $x \sceq{\Sc} \seqempty$.  But $\Sc\ent\seqlen{\seqempty}\teq 0$ by saturation of $\rn{L-Intro}$, so by ($\ast$) and saturation of $\rn{S-Prop}$, this implies $\Ac\ent 0\teq 1$, which contradicts saturation of $\rn{A-Conf}$.  In the second case,
there exists a variable concatenation term $\seqconiii{x_1}{\cdots}{x_n}\in\ter{\Sc}$
such that
$\Sc\ent x\teq \seqconiii{x_1}{\cdots}{x_n}$
and $\seqconiii{x_1}{\cdots}{x_n}$ is not singular in $\Sc$.
By \Cref{lem:concat-term-len}, we know that 
$\Ac\ent[\lian]\Sigma_{i=1}^{n}\ell_{x_n}\ge 2$.
But by ($\ast$) and saturation of $\rn{S-Prop}$
and $\rn{L-Intro}$, 
together with \Cref{assumption},
we also have $\Ac\ent\Sigma_{i=1}^{n}\ell_{x_n}=1$, which also contradicts saturation of $\rn{A-Conf}$.
\end{proof}

\begin{definition}
Define $\Sc\entnf^{1,3} x\teq t$ if there is a derivation of $\Sc\entnf x\teq t$ without using \Cref{it:entnf-con} of \Cref{def:entnf}.
\end{definition}

\begin{lemma}
\label{lem:entnf-two-elim}
If $\Sc\entnf^{\ast} x\teq \vec{y}$, where $y$ has size $n$, then for some $\vec{z}$ of size $n$ such that $z_i\sceq{\Sc}y_i, i\in[1,n]$, $\Sc\entnf^{1,3} x\teq\vec{z}$.
\end{lemma}
\begin{proof}
Since $\Sc\entnf^{\ast} x\teq \vec{y}$, there is some $\vec{z'}$ such that $z'_i\sceq{\Sc}y_i, i\in[1,n]$ and $\Sc\entnf x\teq \vec{z'}$.  Consider the derivation tree described in \Cref{lem:entnf-term}, and let $D$ be the path through the tree corresponding to the derivation of $\Sc\entnf x\teq \vec{z'}$.  
If $D$ has no application of \Cref{it:entnf-con} of \Cref{def:entnf},
then the claim is proved by setting $\vec{z}$ to be $\vec{z'}$.
Otherwise,
the first node in $D$ must use \Cref{it:entnf-con} of \Cref{def:entnf} to derive $\Sc\entnf x\teq t$, where $x\teq t\in\Sc$.  Suppose that $t$ is not singular.  Then, it is possible to derive $\Sc\entnf^{1,3} x\teq t$ by starting with $\Sc\entnf^{1,3} x\teq x$ and then applying \Cref{it:entnf-rec} of \Cref{def:entnf}, using $\Sc\ent x\teq t$.  We can then replace the root of $D$ with this derivation to get a derivation showing $\Sc\entnf^{1,3} x\teq\vec{z}$.

Suppose, on the other hand, that $t$ is singular.  Suppose $t=\seqconiii{t_1}{\dots}{t_m}$.  Without loss of generality, assume that $D$ eagerly applies \Cref{it:entnf-rec} of \Cref{def:entnf} $m-1$ times, each time using $\Sc \ent t_i\teq\seqempty$ for some $i\in[1,m]$, which is possible because $t$ is singular.  The resulting node is a derivation of $\Sc\entnf x\teq t_k$ for some variable $t_k$.  We consider two cases.
\begin{enumerate}
    \item Suppose that $t_k$ is atomic with atomic representative $v$.  Then, clearly we have $\Sc\entnf^{\ast} x\teq v$ and $\Sc\entnf^{\ast} t_k\teq v$, so by saturation of \rn{C-Eq}, $x\sceq{\Sc} t_k$.  By \Cref{lem:nf-unique}, we also have that $\vec{y}=v$.  But then, let $\vec{z} = x$.  Clearly, we have $\Sc\entnf^{1,3}x \teq \vec{z}$.  Furthermore, $x\sceq{\Sc}\vec{y}$, proving the claim.
    \item Suppose that $t_k$ is not atomic.  Then $D$ must continue after $\Sc\entnf x\teq t_k$, and the next step must use \Cref{it:entnf-rec} of \Cref{def:entnf} using $\Sc\ent t_k\teq t'$ to derive $\Sc\entnf x\teq t'$, where $t'$ is $\seqempty$ or a variable concatenation term in $\Sc$ that is not singular in $\Sc$.  But then, note that we can start with $\Sc\entnf^{1,3} t_k\teq t_k$ and apply the same step to get $\Sc\entnf^{1,3} t_k\teq t'$.  If we continue using the rest of the steps in derivation $D$, we can show that $\Sc\entnf^{1,3}t_k\teq\vec{z'}$, and therefore, $\Sc\entnf^{\ast}t_k\teq\vec{y}$.  By saturation of \rn{C-Eq}, we then have $x\sceq{\Sc}t_k$.  But, since $x\sceq{\Sc}t_k$, this means that $x\sceq{\Sc}t'$.  So, we can start with $\Sc\entnf^{1,3} x\teq x$ and apply \Cref{it:entnf-rec} of \Cref{def:entnf} using $\Sc\ent x\teq t'$ to get $\Sc\entnf^{1,3}x\teq t'$.  Using the same steps of that appear in $D$ after $t_k\teq t'$, we can show $\Sc\entnf^{1,3}x\teq\vec{z'}$, which proves the claim.
\end{enumerate}
\end{proof}

\begin{lemma}
\label{lem:roe}
Let $x$ be a sequence variable.
Suppose $\Sc\entnf^{1,3} x\teq t$ and
   $\Sc\entnf z\teq \nf{(\seqconiii{\vec{u}}{x}{\vec{v}})}$.
Then 
$\Sc\entnf z\teq \nf{(\seqconiii{\vec{u}}{t}{\vec{v}})}$.
\end{lemma}

\begin{proof}
By induction on the number of derivation steps in $\entnf^{1,3}$ that yield $\Sc\entnf^{1,3} x\teq t$ (see \Cref{def:entnf}).
If this number is $1$, then it must be by using \Cref{it:entnf-ref} of \Cref{def:entnf}, so $x= t$ and the result follows trivially.
If this number is some $n+1 > 1$, then consider the first $n$ steps of the derivation.
Let $\Sc\entnf^{1,3} x\teq s$ be their result.
By the induction hypothesis, $\Sc\entnf z\teq \nf{(\seqconiii{\vec{u}}{s}{\vec{v}})}$.
Now, consider the $n+1$ step of the derivation. It must replace some variable $y$ in $s$ by some term $r$,
which results in $t$. 
Performing the same step on $\Sc\entnf z\teq \nf{(\seqconiii{\vec{u}}{s}{\vec{v}})}$ results in
$\Sc\entnf z\teq \nf{(\seqconiii{\vec{u}}{t}{\vec{v}})}$.
\end{proof}

\begin{lemma}
\label{lem:concnormnormconc}
Let $x_1,\ldots,x_k$ be sequence variables.
Suppose $\Sc\entnf x\teq \seqconiii{x_1}{\cdots}{x_k}$, and for every $i\in[1,k]$,
$\Sc\entnf^{\ast} x_i\teq \seqconiii{x_i^{1}}{\cdots}{x_i^{n_i}}$.
Then $\Sc\entnf^{\ast} x\teq\seqconiii{x_1^1}{\cdots}{x_k^{n_k}}$.
\end{lemma}

\begin{proof}
We have
$\Sc\entnf x\teq \seqconiii{x_1}{\cdots}{x_k}$.
Also, for every $i\in[1,k]$, 
since $\Sc\entnf^{\ast} x_i\teq \seqconiii{x_i^{1}}{\cdots}{x_i^{n_i}}$,
we also have, by \Cref{lem:entnf-two-elim},
$\Sc\entnf^{1,3} x_i\teq \seqconiii{y_i^1}{\cdots}{y_i^{n_i}}$ for some $y_i^1,\ldots,y_i^{n_i}$
such that $y_i^j\sceq{\Sc}x_i^j$ for every $j\in[1,n_i]$.
Considering the case where $i=1$, by \Cref{lem:roe} we get
$\Sc\entnf x\teq \nf{(\seqconiv{(\seqconiii{y_1^1}{\cdots}{y_1^{n_1}})}{x_2}{\cdots}{x_k})}$. 
Continuing this way until $i=k$, and by the properties of $\nf{}$, we obtain
 $\Sc\entnf x\teq\nf{(\seqconiii{y_1^1}{\cdots}{y_k^{n_k}})}$.
Since $y_1^1,\ldots, y_k^{n_k}$ are variables, we actually have
$\Sc\entnf x\teq\seqconiii{y_1^1}{\cdots}{y_k^{n_k}}$.
And hence, $\Sc\entnf^{\ast} x\teq\seqconiii{x_1^1}{\cdots}{x_k^{n_k}}$.
\end{proof}

Finally, 
we can conclude that $\M\models\Sc$,
by considering each shape of a literal in $\Sc$
separately.

\begin{proof}[Proof of \Cref{lem:MsatSc}]
Let $\varphi\in\Sc$. We prove that $\M\models\varphi$.
By \Cref{assumption}, $\varphi$ is a flat sequence constraint.
We consider the possible shapes of $\varphi$.

\begin{enumerate}
    \item $\varphi$ is $x\teq y$, where $x,y$ have sort $\sint$: using rule $\rn{S-Prop}$, $x\teq y\in \Ac$.
    By \Cref{lem:arith}, $\M\models\varphi$.
    \item $\varphi$ is $x\teq y$, where $x,y$ have sort $\selem$: we know that $x\sceq{\Sc}y$, so by \Cref{def:mc-int-elem}, $\M(x)=\M(y)$.  Thus, $\M\models\varphi$.
    \item $\varphi$ is $x\teq y$ where $x,y$ have sort $\sseq$: \Cref{def:mc-length,def:mc-unit,def:mc-modelcon,def:mc-non-atom} are defined for equivalence classes. Since $x$ and $y$ are in the same equivalence class, $\M\models \varphi$.
    \item $\varphi$ is $x\tneq y$ where $x,y$ have sort $\sint$: by saturation of $\rn{S-Prop}$, $x\teq x,y\teq y\in\Ac$. Hence, $x,y\in\ter{\Sc}\cap\ter{\Ac}$.
    By saturation of $\rn{S-A}$, we either have $x\teq y\in\Ac$ or $x\tneq y\in\Ac$.
    In the first case, by saturation of \rn{A-Prop}, we also have
    $x\teq y\in\Sc$, which is impossible by saturation of $\rn{S-Conf}$.  Hence, we have $x\tneq y\in\Ac$. 
    By \Cref{lem:arith}, $\M\models\varphi$.
    \item $\varphi$ is $x\tneq y$ where $x,y$ have sort $\selem$: $x$ is not equivalent to $y$ w.r.t. $\sceq{\Sc}$, so \Cref{def:mc-int-elem}
    assigns them different values; thus,
    $\M(x)\neq \M(y)$, and hence, $\M\models\varphi$.
    \item $\varphi$ is $x\tneq y$ where $x,y$ have sort $\sseq$: $\M\models\varphi$ by \Cref{lem:disequal}.
    \item $\varphi$ is $x\teq y+z$, $x\teq -y$, or $x\teq n$ for some $n\in\mathbb{N}$: by saturation of \rn{S-Prop}, $\varphi\in\Ac$, so $\M\models\varphi$ by \Cref{lem:arith}.
    \item\label{it:seqempty} $\varphi$ is $x\teq \seqempty$: we know that $\Sc\ent \seqlen{\seqempty}=0$ by saturation of $\rn{L-Intro}$.
    Using \Cref{assumption}, we get
    $\Sc\ent \ell_{x}\teq 0$, and so by saturation of $\rn{S-Prop}$
    we have $\Ac \ent \ell_{x}=0$.  
    It follows by \Cref{lem:seqlen} that $\M(x)$ has length 0 and is thus the empty sequence.
\item $\varphi$ is $x\teq \sequnit(y)$:
    By \Cref{lem:unitatomic}, $[x]$ is atomic.
    Also, $[x]$ is a unit equivalence class. Hence
    $\M([x])$ was defined in \Cref{def:mc-unit} and was 
    set to a sequence of size 1 whose only element is $\M(y)$ by \Cref{lem:unit-wd}. Therefore, $\M\models x\teq \sequnit(y)$.
    \item $\varphi$ is $x\teq \seqlen{y}$: we know that $\ell_y\teq\seqlen{y}\in\Sc$ by \Cref{assumption}, so $\ell_y\teq x\in\Ac$ by saturation of $\rn{S-Prop}$.  From \Cref{def:mc-int-elem}, it follows that $\M(\ell_y)=\M(x)$.  But by \Cref{lem:seqlen}, we also have $\M(\seqlen{y})=\M(\ell_y)$, so $\M(x)=\M(\seqlen{y})$.
    \item\label{it:concat} $\varphi$ is $x\teq\seqconiii{x_1}{\cdots}{x_n}$.
        Suppose that $x_1, \ldots, x_n$ have the unique (by \Cref{lem:nf-unique}) normal forms $\Sc\entnf^\ast x_1\teq \seqconiii{u_1}{\cdots}{u_{m_1}}$, $\Sc\entnf^\ast x_2\teq\seqconiii{u_{m_1+1}}{\cdots}{u_{m_2}}$,
    $\ldots$,
    $\Sc\entnf^\ast x_n\teq\seqconiii{u_{m_{n-1}+1}}{\cdots}{u_{m_n}}$.
    By \Cref{it:entnf-con} of \Cref{def:entnf}, we know $\Sc\entnf x\teq \seqconiii{x_1}{\cdots}{x_n}$, so by \Cref{lem:concnormnormconc}, we have $\Sc\entnf^\ast x\teq \seqconiii{u_1}{\cdots}{u_{m_n}}$.
    By \Cref{lem:model-eq-nf}, then $\M(x)=\seqconiii{\M(u_1)}{\dots}{\M(u_{m_n})}$.  Also, by \Cref{lem:model-eq-nf}, for each $i\in[1,n]$, $\M(x_i)=\seqconiii{\M(u_i)}{\dots}{\M(u_{m_i})}$.  So, $\M(\seqconiii{x_1}{\dots}{x_n})=\seqconiii{\M(x_1)}{\dots}{\M(x_n)}=\seqconiii{\M(u_1)}{\dots}{\M(u_{m_n})}=\M(x)$.
    
    
    \item $\varphi$ is $x\teq \seqnth(y,i)$: We consider two cases:
    \begin{enumerate}
        \item $\M(i)$ is negative or is not smaller than the length of $\M(y)$: applying \Cref{def:mc-nth} with
        $\M(y)$ for $a$, $\M(i)$ for $i$ and $x$ for itself,
        we get that $\M\models \varphi$.
        \item $\M(i)$ is non-negative and smaller than the length of $\M(y)$: By \Cref{lem:nthinrange}
        with $x$ for $k$,
        $y$ for $x$ and $i$ for $y$,
        we have that the $\M(i)$th element of $\M(y)$ is $\M(x)$.
        Therefore, $\M\models \varphi$.
    \end{enumerate}
    
    \item $\varphi$ is $x\teq\sequpdate(y,i,z)$: 
    First, assume $\M(i)$ is negative or not smaller than the length of $\M(y)$.
    In this case, the interpretation in $\M$ of
    $\sequpdate(y,i,z)$ is $\M(y)$. Hence, we prove that
    $\M(x)=\M(y)$.
    By saturation of \rn{Update-Bound}, we have that
    either $\Ac\ent 0\leq i<\ell_y$ or $x\teq y\in\Sc$.
    The first case is impossible by \Cref{def:mc-int-elem,def:mc-fun} and \Cref{lem:seqlen},
    and hence, the second case holds.
    We therefore have $x\sceq{\Sc}y$, and since
    the definitions of sequence variables are done by equivalence classes (see \Cref{def:mc-unit,def:mc-modelcon,def:mc-non-atom}), we have
    $\M(x)=\M(y)$.
    
    Next, assume $\M(i)$ is non-negative and smaller than
    the length of $\M(y)$ (this also implies that $\M(y)$ is not an empty sequence).
    By \Cref{lem:eq_len}, we have $\M(\ell_x)=\M(\ell_y)$, and by \Cref{lem:seqlen}, $\M(\seqlen{x})=\M(\ell_x)$ and $\M(\seqlen{y})=\M(\ell_y)$.
    We consider the following cases:
    \begin{enumerate}
        \item\label{it:update-atomic} $[x]$ and $[y]$ are atomic in $\Sc$:
        We show that for every $j\in[0,\M(\ell_x))$,
        the $j$th element of $\M(x)$ is the same as
        the $j$th element of $\M(\sequpdate(y,i,z))$.

        First, suppose $j=\M(i)$.
        We show that the $j$th element of $\M(x)$ is $\M(z)$.
        By saturation of $\rn{Nth-Intro}$,
        we have $\seqnth(x,i)\in\ter{\Sc}$.
        Since $x\teq\sequpdate(y,i,z)\in\Sc$,
        rule $\rn{Nth-Update}$ applies.
        Notice that its first and last branches are impossible:
        the first by our assumption on $\M(i)$ and \Cref{def:mc-int-elem,def:mc-fun} and \Cref{lem:seqlen}, and the last because it would require $\Ac\ent i\neq i$, but we know $\Ac\notent\bot$ by saturation of \rn{A-Conf}.
        Hence, the middle branch applies, which means
        $\seqnth(x,i)\teq z\in\Sc$.  By \Cref{lem:nthinrange-atomic}, then, we know that the $j$th element of $\M(x)$ is $\M(z)$.

        Suppose, on the other hand, that $j\neq\M(i)$.  This implies $\M(\ell_x)=\M(\ell_y)\ge 2$, so $[x]$ and $[y]$ cannot be unit by \Cref{def:mc-unit}.  Their values are thus set by \Cref{def:mc-modelcon}.
        Since $x\teq\sequpdate(y, i, z)\in\Sc$, by \Cref{def:mc-weak-equiv}, there is an edge $d$ between $[x]$ and $[y]$ with $\M(i) \in \delta(d)$.
        It follows, because $j\neq\M(i)$, that $[x]\sim_j [y]$ by \Cref{def:weak_equiv_rel}.  But then $[x]$ and $[y]$ are in the same equivalence class of $\sim_j$, so in either case of 
        \Cref{def:mc-modelcon}, their $j$th value is set to the same value.  Finally, because $j\neq\M(i)$, the $j$th element of $\M(\sequpdate(y,i,z))$ is (according to the semantics of $\sequpdate$) the $j$th element of $\M(y)$.

        
        \item Suppose $\Sc\entnf^\ast y \teq \seqconiii{w_1}{\cdots}{w_n}$.  By \Cref{lem:update-align}, we have
        $x\teq \seqconiii{z_1}{\cdots}{z_n}\in\Sc$ for some atomic $z_1,\ldots,z_n$, $\sum_{j=1}^{k-1} \M(\ell_{w_j}) \leq \M(i) < \sum_{j=1}^k \M(\ell_{w_j})$ and
        $z_k=\sequpdate(w_k, \alpha_k, z)\in\Sc$ for some $k\in[1,n]$, where
        $\Ac\ent \alpha_k\teq i-\sum_{j=1}^{k-1}\ell_{w_j}$, and for $m\in[1,n], m\neq k$, $z_m\teq w_m\in \Sc$.
        Since \Cref{def:mc-unit,def:mc-modelcon} assign equivalence classes, we also have $\M(z_m)=\M(w_m)$.
        
        Since $z_k, w_k$ are atomic, we have
        $\M(z_k)=\M(\sequpdate(w_k, \alpha_k, z))$ by \Cref{it:update-atomic} above.
        By \Cref{lem:arith}, we have $\M(\alpha_k) = \M(i)-\sum_{j=1}^{k-1}\M(\ell_{w_j})$.
        By \Cref{it:concat}, above, $\M(x)\teq\seqconiii{\M(z_1)}{\dots}{\M(z_n)}$. It follows that $\M(x)=\seqconv{\M(w_1)}{\dots}{\M(\sequpdate(w_k,\alpha_k,z))}{\dots}{\M(w_n)}$.  By \Cref{lem:atomic-nf}, we also have $\M(y)=\seqconiii{\M(w_1)}{\dots}{\M(w_n)}$, so $\M(\sequpdate(y,i,z))=\seqconv{\M(w_1)}{\dots}{\M(\sequpdate)(\M(w_k),\M(i)-\sum_{j=1}^{k-1}\M(\ell_{w_j}),\M(z))}{\dots}{\M(w_n)}$.
        It follows that $\M(x)=\M(\sequpdate(y, i, z))$.
    \end{enumerate}
    
    \item $\varphi$ is $x\teq\seqextract(y,i,j)$: By saturation w.r.t. $\rn{R-Extract}$, we have two options.
    
    In the first, 
    $\Ac\ent i<0\vee i\geq\ell_{y}\vee j\leq 0$ and $x\teq\seqempty\in\Sc$.
    By \Cref{it:seqempty}, above,
    $\M(x)$ is the empty sequence.  By \Cref{lem:arith,lem:seqlen},
    $\M(i)<0$ or $\M(i)$ has at least the length of $\M(y)$ or $\M(j)\leq 0$.
    In each of these 3 cases, we get from
    the semantics of $\seqextract$ in $\sth$ that
    $\M(\seqextract)$ assigns the empty sequence
    w.r.t the inputs $\M(y)$, $\M(i)$ and $\M(j)$.
    Hence in this case we get $\M\models\varphi$.
    
    In the second case, 
    $\Ac\ent 0\leq i < \ell_{y} \wedge j>0 \wedge \ell_{k}\teq i \wedge \ell_{x}\teq\min(j,\ell_{y}-i) \wedge \ell_{k'}\teq \ell_y-\ell_x-i$
    and $y\teq\seqconiii{k}{x}{k'}\in\Sc$.
    We thus have $\M(y)=\seqconiii{\M(k)}{\M(x)}{\M(k')}$ by \Cref{it:concat}, above. Also, $\M\ent\Ac$, by \Cref{lem:arith}.
    According to the semantics of $\seqextract$ in $\sth$,
    since $\M(i),\M(j)\geq 0$, $\M(i)$ is smaller than
    the length of $\M(y)$,
    the value assigned in $\M$ to
    $\seqextract(y,i,j)$ is the maximal sub-sequence 
    of $\M(y)$ that starts at index $\M(i)$ and has
    length at most $\M(j)$.
    Since $\M\models y\teq\seqconiii{k}{x}{k'}$ with the appropriate length constraints (by \Cref{lem:seqlen}),
    we have that this sequence value is exactly $\M(x)$
    and hence, $\M\models\varphi$.
    
\end{enumerate}
\end{proof}
\end{report}

\end{document}